\documentclass[journal,twoside]{IEEEtran}

\IEEEoverridecommandlockouts
\usepackage{cite}
\usepackage{amsmath,amssymb,amsfonts,amsthm}
\usepackage{algorithmic}
\usepackage{parcolumns}
\usepackage{graphicx}
\usepackage{textcomp}
\usepackage{xcolor}
\usepackage{gensymb}
\usepackage{tabularx}
\usepackage{mathtools}
\newcommand\numberthis{\addtocounter{equation}{1}\tag{\theequation}}
\newtheorem*{result}{Theorem 1}
\belowdisplayskip=\abovedisplayskip
\def\BibTeX{{\rm B\kern-.05em{\sc i\kern-.025em b}\kern-.08em
		T\kern-.1667em\lower.7ex\hbox{E}\kern-.125emX}}

\begin{document}
	
\bstctlcite{IEEEexample:BSTcontrol}

\title{Optimal Phase Design for RIS Channel Estimation}

\author{Chelsea~L.~Miller,
	    Peter~J.~Smith,
	    Pawel~A.~Dmochowski
	\thanks{C.~L.~ Miller was with the School of Engineering and Computer Science, %
	        Victoria University of Wellington, New Zealand. She is currently at Spark New Zealand, Wellington, New Zealand %
	        (e-mail: chelsea.miller@ecs.vuw.ac.nz).}%
	\thanks{P.~J.~Smith is with the School of Mathematics and Statistics, %
	        Victoria University of Wellington, Wellington, New Zealand %
	        (e-mail: peter.smith@vuw.ac.nz).}%
	\thanks{P.~A.~Dmochowski is with the School of Engineering and Computer Science, %
	        Victoria University of Wellington, Wellington, New Zealand. %
	        (e-mail: pawel.dmochowski@vuw.ac.nz.}%
	}
\markboth{Submitted to IEEE Transactions on Wireless Communications}{Miller \MakeLowercase{\textit{et al.}}: Optimal Phase Design for RIS Channel Estimation}
	
	\maketitle
	
	\begin{abstract}
		We develop an optimal version of a prior two-stage channel estimation protocol for RIS-assisted channels. The new design uses a modified DFT matrix (MDFT) for the training phases at the RIS and is shown to minimize the total channel estimation error variance. In conjunction with interpolation (estimating fewer RIS channels), the MDFT approach accelerates channel estimation even when the channel from base station to RIS is line-of-sight. In contrast, prior two-stage techniques required a full-rank channel for efficient estimation. We investigate the  resulting channel estimation errors by comparing different training phase designs for a variety of propagation conditions using a ray-based channel model. To examine the overall performance, we simulate the spectral efficiency with MRC processing for a single-user RIS-assisted system using an existing optimal design for the RIS transmission phases. Results verify the optimality of MDFT while simulations and analysis show that the performance is more dependent on the user-to-RIS channel correlation and the coarseness of the interpolation used, rather than the training phase design. For example, under a scenario with more highly correlated channels, the procedure accelerates channel estimation by a factor of 16, while the improvement is a factor of 5 in  a less correlated case. The overall procedure is extremely robust, with a maximum performance loss of 1.5bits/sec/Hz compared to that with perfect channel state information for the considered channel conditions.
	\end{abstract}
	

\section{Introduction}\label{intro}

Reconfigurable intelligent surfaces (RIS) have become increasingly popular in the wireless community due to their potential to enhance performance without additional RF chains, reducing the added complexity and power consumption\cite{Mishra2019Channel}. A typical RIS comprises a large number of adjacent reflecting surfaces, each of which is independently controlled by a base station (BS) to reflect incident rays with a desired phase shift. This paper considers an example of a cellular wireless application, wherein the RIS assists communication between a multi-antenna BS and a user (UE). Used in this way, the RIS can be used to boost data rate, reliability or support more users. A common use case is  to overcome the high attenuation and blockage of millimeter wave communications. 
	
To inform the configuration of RIS phases to enhance data transmission, the BS must  acquire channel state information (CSI) for the UE-to-RIS link (UE-RIS), the RIS-to-BS link (RIS-BS), and the direct UE-to-BS link (UE-BS). Once the RIS phases are set, the BS must also use an estimate of the total channel, combining all of the aforementioned channels, to perform signal processing. Many early publications\cite{Nadeem2020Intelligent,Mishra2019Channel} estimate the cascaded channel from UE to BS through the RIS (UE-RIS-BS), and some additionally estimate the direct UE-BS channel \cite{Mishra2019Channel,Nadeem2020Intelligent} while others assume it is blocked\cite{Bjornson2020Intelligent}. For $M$ antennas at the BS and $N$ elements at the RIS, a minimum of $MN$ cascaded channel coefficients must be estimated with an additional $M$ for the direct channel.
	
This estimation can be performed by setting the RIS phases using a series of known phase vectors from a training phase matrix. For example, the columns of an $N\times N$ identity matrix (equivalent to measuring the channel for one active RIS element per training symbol)\cite{Mishra2019Channel,Jensen2020Optimal} or a DFT matrix\cite{Wan2020Broadband,Jensen2020Optimal,Nadeem2020Intelligent,Zheng2020Intelligent}. The authors of \cite{Jensen2020Optimal} find that the DFT training phase matrix significantly reduces the channel estimation (CE) error compared to the identity matrix used in \cite{Mishra2019Channel}, while the authors of \cite{Nadeem2020Intelligent} and \cite{Zheng2020Intelligent} further claim that the DFT matrix provides an optimal training phase design for the cascaded channel. However, little analytical progress has been made in the literature to support these claims. Estimation of the full cascaded channel requires a minimum of $N+1$ training symbols using the methods in \cite{Zheng2020Intelligent,Nadeem2020Intelligent}. As $N$ is typically on the order of $10^2$\cite{Bjornson2020Intelligent}, this consumes a large portion of the transmission frame duration, significantly reducing the number of usable frame symbols for data transmission. Furthermore, the cascaded channel combines the channels for the UE-RIS and RIS-BS link, preventing the use of RIS phase designs requiring separate knowledge of these two links, such as the optimal single-user phase design in \cite{Singh2020Optimal}.
	
The RIS-BS and RIS-UE channels are separately estimated in the CE protocol proposed in \cite{Hu2019TwoTimescale}. The RIS-BS channel is assumed to be quasi-static, and hence the estimate of this link is iteratively refined over a large timescale. Each RIS-BS channel estimate is used to estimate the UE-RIS and UE-BS channels using a training phase matrix comprised of uniform random phases. This method reduces the number of pilots required for the estimation of the UE-RIS and UE-BS channels by a factor of $1/M$. However, this reduction in pilot overhead relies on the RIS-BS channel being full-rank (rank-$M$, as we assume $M<N$).
	
In this paper, we consider the scenario where the RIS is placed in close proximity to the BS, and above the clutter\cite{Nadeem2020Intelligent}. This results in a RIS-BS channel which is dominated by a strong, static line-of-sight (LoS) component, with minimal variable scattering \cite{Liu2019MatrixCalibration,Wan2020Broadband,Nadeem2020asymptotic}. The static LoS component can be estimated offline at the BS\cite{Wan2020Broadband} and the RIS-BS link may be approximated using only this component. This approximation eliminates the need for the additional pilots used to estimate the BS-RIS channel in \cite{Hu2019TwoTimescale}. Unfortunately the LoS component is also rank-1\cite{Wan2020Broadband} and does not support the method in \cite{Hu2019TwoTimescale}. This necessitates a suitable method of reducing pilot overhead for rank-deficient channels.
	
Another method of pilot reduction involves grouping RIS elements into sub-surfaces comprising a number of adjacent elements which are assumed to have similar channels \cite{You2020Intelligent,Zheng2020Intelligent}. The pilot overhead is reduced because the channel for only one ``active'' RIS element per group must be estimated, requiring one pilot symbol per group. The accuracy of the resulting CSI is heavily dependent on the correlation between channel coefficients in each group, which is affected by several aspects of the RIS geometry and channel statistics. Secondly, grouping introduces CSI error which cannot be removed through more accurate measurement of the active elements and must be addressed by other means, as  in \cite{miller2021efficient} where a second estimation stage is employed.
	
The authors of \cite{miller2021efficient} propose a two-stage estimation procedure where the BS-RIS channel comprises a known LoS component and unknown scattering. In Stage 1, the known LoS component is used to perform LMMSE estimation similar to that in \cite{Hu2019TwoTimescale} using DFT training phases, providing an estimate of the UE-BS channel, and the UE-RIS channel coefficients for a subset of active RIS elements. These elements are used to interpolate the complete UE-RIS channel using three interpolation methods, and the RIS phases for data transmission are configured using the estimated CSI. The channel error introduced from grouping and interpolation is addressed with a second estimation stage. Stage 2 uses additional pilot symbols to re-estimate the combined channel, providing a more accurate estimate for UL processing during data transmission. The procedure is found to be very robust, even with significant unknown scattering in the RIS-BS link. However, like much of the literature to date, this study does not provide any analysis of the effects of the training phase matrix on the channel estimation error to facilitate improved training phase design.
	
In light of the gaps in the current literature, we adopt the RIS-assisted channel estimation procedure from \cite{miller2021efficient} and build on this work by providing an analysis of the effects of training phases on the channel estimation error in Stage 1. We use this to design an optimal training phase matrix, and examine the estimation error under various propagation conditions in comparison to alternative designs. We focus on the single-user case, noting that many existing multi-user channel estimation procedures for RIS-assisted transmissions are simply extensions or repetitions of a single-user estimation procedure\cite{Nadeem2020Intelligent,wang2019channel,Zheng2020Intelligent}. To provide an example of the impact of CE on system performance, we configure the RIS phases according to the design in \cite{Singh2020Optimal}, and examine the spectral efficiency (SE) with UL maximal ratio combining (MRC) using the re-estimated total channel from Stage 2 of the estimation procedure. The RIS phases in \cite{Singh2020Optimal} and MRC processing are both optimal for the single-user case and can be expressed in closed form. Hence, this is an ideal scenario to use in evaluating the effects of CE. More specifically, the contributions of this paper are as follows:
\begin{itemize}
    \item We analyse the variance of the error in the UE-BS and UE-RIS channel estimates from Stage 1, identifying the error components arising from unknown scattering in the RIS-BS channel and processed noise at the BS.
	\item We show that the error from the unknown scattering is independent of the training phase vectors and therefore cannot be reduced through training phase design.
	\item We develop an optimal training phase design, labelled MDFT, based on a modification of the DFT matrix.
	\item We examine the normalised mean-squared error (NMSE) of the total channel estimate after Stage 1, the resulting RMS phase error in the RIS transmission phases, and the SE after re-estimation in Stage 2. We compare the performance of several interpolating channel estimation methods with optimal (MDFT), DFT, and random training phase matrices under several conditions. Results show that the robust channel estimation procedure provides remarkable performance unless coarse interpolating methods are used in channels with very low correlation. 
\end{itemize}
%
%
\section{System Model}\label{s:system_model}
We consider a cell with radius $r$ and exclusion radius $r_0$ centered on the origin of the $x$-$y$ azimuth plane. A BS equipped with $M$ omni-directional antennas arranged in a uniform linear array (ULA) with inter-element spacing $d^{\mathrm{\scriptscriptstyle B}}$ wavelengths communicates with a RIS of $N$ reflecting elements arranged in a vertical uniform rectangular array (VURA) with $N_y$ columns of $N_z$ elements with $N_yN_z = N$ and spacing $d^\mathrm{\scriptscriptstyle R}$ wavelengths. The BS and RIS are positioned within the exclusion zone at a distance $D^\mathrm{\scriptscriptstyle BR}$ apart on the $y$-axis, rotated in the azimuth plane such that array broadside is at an angle $\pi/2-\bar{\phi}^{\mathrm{\scriptscriptstyle B}}$ and $-\pi/2-\bar{\phi}^\mathrm{\scriptscriptstyle R}$ relative to the $x$-axis, respectively (see Fig.\ref{fig:RIS_diagram}(a)). We consider the single-user case, noting that the resulting CE procedure may be repeated during the training phase for the multi-user case. The single-antenna UE is placed outside the exclusion zone on the $x$-axis at a distance $D^\mathrm{\scriptscriptstyle U}$ from the origin with $r_0<D^\mathrm{\scriptscriptstyle U}<r$, resulting in a distance $D^\mathrm{\scriptscriptstyle BU}$ from UE to BS and $D^\mathrm{\scriptscriptstyle RU}$ from UE to RIS.
\begin{figure}[ht]
	\centering
	\includegraphics[trim=1cm 0.04cm 1.6cm 0.1cm, clip=true, width=1\columnwidth]{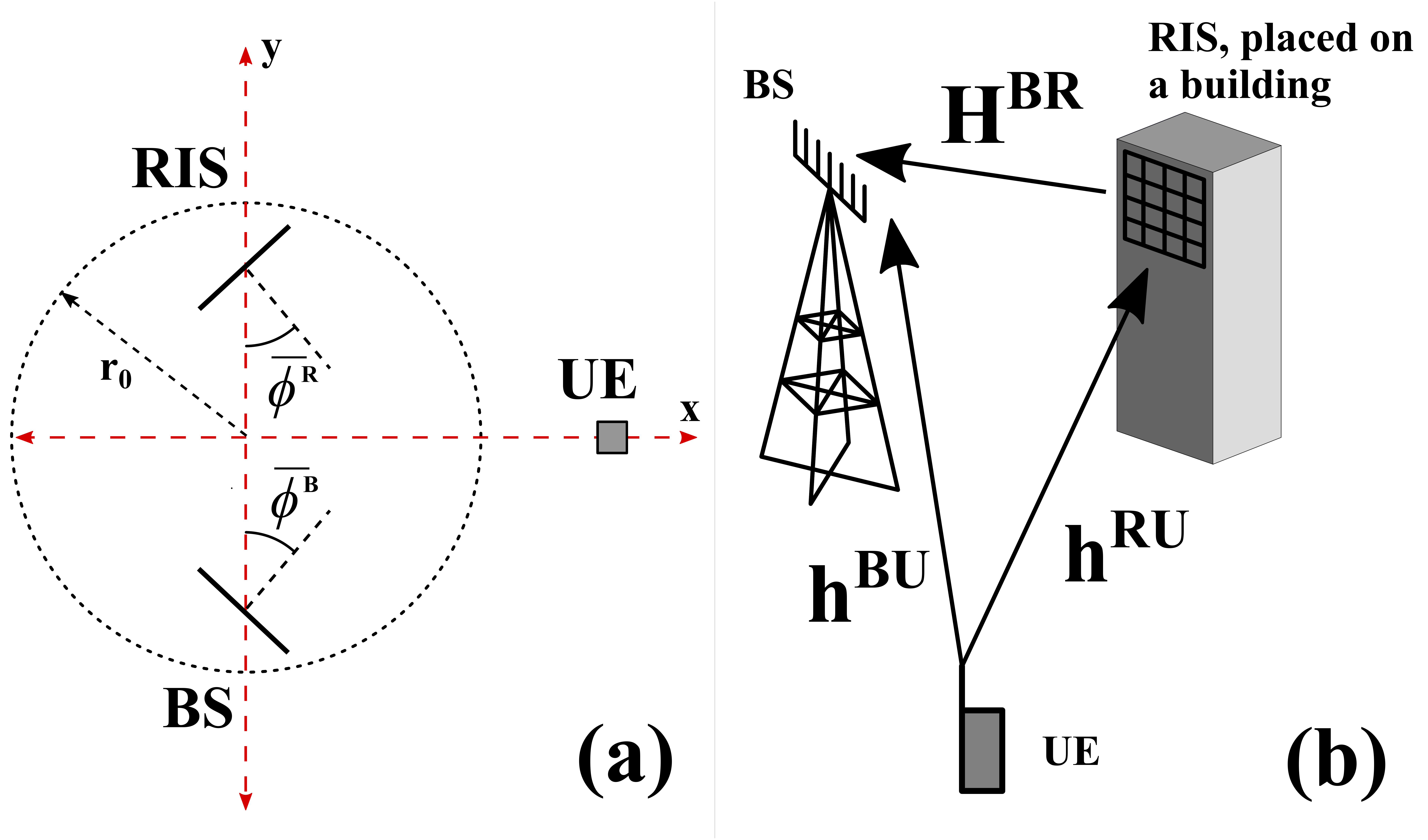}
	\caption{Bird's eye view of exclusion zone (a), and system diagram (b).}
	\label{fig:RIS_diagram}
\end{figure}
We model large-scale fading variables, $\beta$, using the classic pathloss equation
\begin{align}\label{eq:pathloss}
\beta^\mathrm{\scriptscriptstyle BU} = 
A^\mathrm{\scriptscriptstyle BU} L^\mathrm{\scriptscriptstyle BU} \bigg(\frac{D^\mathrm{\scriptscriptstyle BU}}{D_0}\bigg)^{-\Gamma},
\end{align}
and similarly for $\beta^\mathrm{\scriptscriptstyle RU}$ and $\beta^\mathrm{\scriptscriptstyle BR}$.  For all such values, $D$ is the distance from the source to destination, $A$ is the unitless attenuation constant representing the attenuation in the absence of shadowing  at reference distance $D_0$ from the source, $10\log_{10}(L) \sim \mathcal{N}(0,\sigma_\mathrm{sf}^2)$ models lognormal shadow fading,  and $\Gamma$ is the pathloss exponent. 

We model the UE-BS and UE-RIS channels, $\mathbf{h}^\mathrm{\scriptscriptstyle BU} \in \mathbb{C}^\mathrm{M \times 1}$ and $\mathbf{h}^\mathrm{\scriptscriptstyle RU} \in \mathbb{C}^\mathrm{N \times 1}$, using a clustered ray-based model. Each channel is composed of a summation of subrays scattered from several clusters\cite{sangodoyin_cluster_2018} so that
\begin{align}\label{eq:ray-based}
	\mathbf{h}^\mathrm{\scriptscriptstyle BU} = \sum_{c=1}^{C^\mathrm{\scriptscriptstyle BU}}\sum_{s=1}^{S^\mathrm{\scriptscriptstyle BU}}\gamma^\mathrm{\scriptscriptstyle BU}_{c,s}\mathbf{a}^{\mathrm{\scriptscriptstyle B}}(\phi^\mathrm{\scriptscriptstyle BU}_{c,s},\theta^\mathrm{\scriptscriptstyle BU}_{c,s}),
\end{align}
and similarly for $\mathbf{h}^\mathrm{\scriptscriptstyle RU}$, for which the above BS steering vectors $\mathbf{a}^{\mathrm{\scriptscriptstyle B}}(\cdot)$ are replaced with RIS steering vectors $\mathbf{a}^\mathrm{\scriptscriptstyle R}(\cdot)$. For all such channels, $C$ is the number of clusters, $S$ is the number of subrays, $\gamma_{c,s} = \sqrt{\beta_{c,s}}\exp(-j\Theta_{c,s})$ incorporates the ray power $\beta_{c,s}$ and uniform random phase offset $\Theta_{c,s} \sim \mathcal{U}[0,2\pi]$. The ray powers are modelled as $\beta_{c,s} = \beta_c/S$, where the cluster powers, $\beta_c$, are set to decay exponentially from the strongest to weakest scattering cluster with ratio $\eta = \beta_{C}/\beta_1$ such that $\sum_{c=1}^{C}\beta_{c} = \beta$\cite{3GPP}. The steering vectors, $\mathbf{a}^{\mathrm{\scriptscriptstyle B}}(\cdot)$ and $\mathbf{a}^\mathrm{\scriptscriptstyle R}(\cdot)$, are determined by the structures of the BS and RIS arrays, and are given by \cite{antenna_config_paper}
\begin{align*}
	\mathbf{a}^{\mathrm{\scriptscriptstyle B}}(\phi,\theta) &=
	[1, e^{-j2\pi d^{\mathrm{\scriptscriptstyle B}}\sin\theta\sin\phi}...
	    e^{-j2\pi d^{\mathrm{\scriptscriptstyle B}}(M-1)\sin\theta\sin\phi}]^\mathrm{T} \\
	\mathbf{a}_y(\phi,\theta) &=
	[1, e^{-j2\pi d^\mathrm{\scriptscriptstyle R}\sin\theta\sin\phi}...
	    e^{-j2\pi d^\mathrm{\scriptscriptstyle R}(N_y-1)\sin\theta\sin\phi}]^\mathrm{T} \\
	\mathbf{a}_z(\theta) &=
	[1, e^{-j2\pi d^\mathrm{\scriptscriptstyle R}\cos\theta}...
	    e^{-j2\pi d^\mathrm{\scriptscriptstyle R}(N_z-1)\cos\theta}]^\mathrm{T}\\
	\mathbf{a}^\mathrm{\scriptscriptstyle R}(\phi,\theta) &= \mathbf{a}_y(\phi,\theta)\otimes\mathbf{a}_z(\theta).
	    \numberthis\label{eq:steeringvectordefinitions}
\end{align*}
The azimuth angle of each ray relative to array broadside, $\phi_{c,s} = \phi_c + \Delta_{c,s}$, is a summation of the random cluster central angle, $\phi_c$, identical for all subrays in cluster $c$, and a random subray offset $\Delta_{c,s}$. Equivalently, in elevation, $\theta_{c,s} = \theta_c + \delta_{c,s}$ is measured from the zenith. 

The RIS-BS channel is modelled similarly to that in \cite{Liu2019MatrixCalibration} as
\begin{align}
\mathbf{H}^\mathrm{\scriptscriptstyle BR} = \tilde{\mathbf{H}}^\mathrm{\scriptscriptstyle BR} + \bar{\mathbf{H}}^\mathrm{\scriptscriptstyle BR},
\end{align}
where the known specular component is given by 
\begin{align*}
\bar{\mathbf{H}}^\mathrm{\scriptscriptstyle BR} = 
\sqrt{\bar{\beta}^\mathrm{\scriptscriptstyle BR}}
\mathbf{a}^{\mathrm{\scriptscriptstyle B}}
(\bar{\phi}^{\mathrm{\scriptscriptstyle B}},
 \bar{\theta}^{\mathrm{\scriptscriptstyle B}})
 \mathbf{a}^{\mathrm{\scriptscriptstyle R}\dagger}
 (\bar{\phi}^\mathrm{\scriptscriptstyle R},
  \bar{\theta}^\mathrm{\scriptscriptstyle R}),
\numberthis\label{eq_LosHbr}
\end{align*}
where we use $\dagger$ to signify the Hermitian transpose, $\bar{\phi}^{\mathrm{\scriptscriptstyle B}}$ and $\bar{\theta}^{\mathrm{\scriptscriptstyle B}}$ are the azimuth and elevation LoS angles of arrival at the BS and $\bar{\phi}^\mathrm{\scriptscriptstyle R}$, $\bar{\theta}^\mathrm{\scriptscriptstyle R}$ are the equivalent LoS angles of departure at the RIS. The RIS-BS link is assumed to have a strong LoS component. The weaker, unknown scattered component is modelled, similar to \eqref{eq:ray-based}, as 
\begin{align}\label{eq:ray-based_BR}
 \tilde{\mathbf{H}}^\mathrm{\scriptscriptstyle BR} = 
 \sum_{c=1}^{C^{\mathrm{\scriptscriptstyle BR}}}
 \sum_{s=1}^{S^{\mathrm{\scriptscriptstyle BR}}}
 \gamma^\mathrm{\scriptscriptstyle BR}_{c,s} 
 \mathbf{a}^{\mathrm{\scriptscriptstyle B}} 
 (\phi^{\mathrm{\scriptscriptstyle B}}_{c,s},
  \theta^\mathrm{\scriptscriptstyle B}_{c,s})
 \mathbf{a}^{\mathrm{\scriptscriptstyle R}\dagger}
 (\phi^\mathrm{\scriptscriptstyle R}_{c,s},
  \theta^\mathrm{\scriptscriptstyle R}_{c,s}),
\end{align}
with $\gamma^\mathrm{\scriptscriptstyle BR}_{c,s} = \sqrt{\tilde{\beta}_{c,s}^\mathrm{\scriptscriptstyle BR}}\exp(-j\Theta_{c,s}^\mathrm{\scriptscriptstyle BR})$. The RIS-BS pathloss is divided among the LoS and scattered component as $\bar{\beta}^\mathrm{\scriptscriptstyle BR} = K\beta^\mathrm{\scriptscriptstyle BR}/(K+1)$ and $\tilde{\beta}^\mathrm{\scriptscriptstyle BR} = \beta^\mathrm{\scriptscriptstyle BR}/(K+1)$, respectively. $K = \bar{\beta}^\mathrm{\scriptscriptstyle BR}/\tilde{\beta}^\mathrm{\scriptscriptstyle BR}$ gives the ratio of power in the specular compared to the scattered component. The global channel is given by
\begin{align}
\mathbf{h}^\mathrm{\scriptscriptstyle TOT} =
\mathbf{H}^\mathrm{\scriptscriptstyle BR}
\boldsymbol{\bar{\Phi}}\mathbf{h}^\mathrm{\scriptscriptstyle RU} +
\mathbf{h}^\mathrm{\scriptscriptstyle BU},
\label{eq:htot_true}
\end{align}
where $\boldsymbol{\bar{\Phi}}$ is the diagonal matrix of RIS phases used during data transmission, designed using the estimated CSI.

%
\section{Channel Estimation and Uplink Beamforming}\label{s:ch_est_beamforming}
We consider uplink transmission throughout a coherence interval of $T$ symbols comprising $\tau$ pilot symbols used for channel estimation, followed by $T-\tau$ information symbols. We follow the two-stage estimation procedure detailed in \cite{miller2021efficient}, which divides the training symbols into $\tau = \tau_1 + \tau_2$.
%
\subsection{Stage 1: Estimate CSI for RIS Phase Calibration}\label{ss:stage_1}
Here, we describe Stage 1 of the procedure in \cite{miller2021efficient}, in which LMMSE channel estimation is used to estimate $\mathbf{h}^\mathrm{\scriptscriptstyle BU}$ and $\mathbf{h}^{\mathrm{\scriptscriptstyle RU}'}$. The latter contains the UE-RIS channel for a subset of $N'\leq N$ \textit{active} RIS elements. Using the VURA RIS structure, the procedure in \cite{miller2021efficient} groups elements within each column of the VURA. Due to the relatively small variance of ray angles in elevation, for many environments the channel elements are much more correlated along columns of the VURA than along rows. The procedure therefore estimates the channel for a number of active elements per column, which are used to interpolate the remaining channel elements in the column using a range of methods. 

To estimate the UE-BS channel and the UE-RIS channel for the \textit{active} RIS elements, the UE sends $\tau_1$ pilot symbols. The $\tau_1$ received signals at the BS are stacked, giving the $M\tau_1 \times 1$ received vector 
\begin{align*}
\mathbf{r} &= 
\sqrt{\rho}\begin{bmatrix}
\mathbf{H}^{\mathrm{\scriptscriptstyle BR}'}\boldsymbol{\Psi}_1 s_1 & \mathbf{I}_M s_1\\
\vdots&\vdots\\
\mathbf{H}^{\mathrm{\scriptscriptstyle BR}'}\boldsymbol{\Psi}_{\tau_1} s_{\tau_1}& \mathbf{I}_M s_{\tau_1}
\end{bmatrix}
\begin{bmatrix}
\mathbf{h}^{\mathrm{\scriptscriptstyle RU}'}\\
\mathbf{h}^\mathrm{\scriptscriptstyle BU}
\end{bmatrix}
+ \mathbf{n}\\
&\triangleq \mathbf{V}\begin{bmatrix}
\mathbf{h}^{\mathrm{\scriptscriptstyle RU}'}\\
\mathbf{h}^\mathrm{\scriptscriptstyle BU}
\end{bmatrix} + \mathbf{n},\numberthis\label{eq:yvec}
\end{align*}
where $\rho$ is the uplink training SNR and $\mathbf{n} \sim \mathcal{CN}(0,\mathbf{I}_\mathrm{M\tau_1})$ is the additive noise at the BS. $\mathbf{H}^{\mathrm{\scriptscriptstyle BR}'}$ contains the channel coefficients of $\mathbf{H}^{\mathrm{\scriptscriptstyle BR}}$ corresponding to the \textit{active} RIS elements. The diagonal matrices, $\boldsymbol{\Psi}_t = \mathrm{diag}(\boldsymbol{\psi}_t)$, contain the RIS training phase vector, $\boldsymbol{\psi}_t$, for $t= 1\dots\tau_1$, taken from the columns of the $N'\times\tau_1$ training phase matrix, $\boldsymbol{\Phi}$. The design of $\boldsymbol{\Phi}$ is discussed in Sec.~\ref{s:design} Noting that the pilot symbols $[s_1\dots s_{\tau_1}]$ sent from the UE simply scale the channel measurements, we assume without loss of generality that $s_t = 1$ for $t = 1\dots\tau_1$.

The estimate of $\mathbf{V}$ at the BS uses the known columns of the LoS component, $\bar{\mathbf{H}}^{\mathrm{\scriptscriptstyle BR}}$, corresponding to the active elements. Hence
\begin{align}
\hat{\mathbf{V}} = \sqrt{\rho}\begin{bmatrix}
\bar{\mathbf{H}}^{\mathrm{\scriptscriptstyle BR}'}\boldsymbol{\Psi}_1 & \mathbf{I}_M\\
\vdots\\
\bar{\mathbf{H}}^{\mathrm{\scriptscriptstyle BR}'}\boldsymbol{\Psi}_{\tau_1} & \mathbf{I}_M
\end{bmatrix}.\label{eq:Vhat}
\end{align}
This is used to perform linear minimum-mean-squared-error CE, 
\begin{align}
\begin{bmatrix}
\hat{\mathbf{h}}^{\mathrm{\scriptscriptstyle RU}'}\\
\hat{\mathbf{h}}^\mathrm{\scriptscriptstyle BU}
\end{bmatrix} = (\hat{\mathbf{V}}^\dagger\hat{\mathbf{V}})^{-1}\hat{\mathbf{V}}^\dagger\mathbf{r}.\label{eq:ZF_CE}
\end{align}
The corresponding channel estimates are given by
\begin{align*}
\begin{bmatrix}
\hat{\mathbf{h}}^{\mathrm{\scriptscriptstyle RU}'}\\
\hat{\mathbf{h}}^\mathrm{\scriptscriptstyle BU}
\end{bmatrix} &=
\begin{bmatrix}
\mathbf{h}^{\mathrm{\scriptscriptstyle RU}'}\\
\mathbf{h}^\mathrm{\scriptscriptstyle BU}
\end{bmatrix} + 
(\hat{\mathbf{V}}^\dagger\hat{\mathbf{V}})^{-1}\hat{\mathbf{V}}^\dagger\tilde{\mathbf{V}}\begin{bmatrix}
\mathbf{h}^{\mathrm{\scriptscriptstyle RU}'}\\
\mathbf{h}^\mathrm{\scriptscriptstyle BU}
\end{bmatrix}\\ 
&+ (\hat{\mathbf{V}}^\dagger\hat{\mathbf{V}})^{-1}\hat{\mathbf{V}}^\dagger\mathbf{n}\\
& \triangleq
\begin{bmatrix}
 \mathbf{h}^{\mathrm{\scriptscriptstyle RU}'}\\
 \mathbf{h}^\mathrm{\scriptscriptstyle BU}
\end{bmatrix} +
\mathbf{\epsilon}_1 + \mathbf{\epsilon}_2,
\numberthis\label{eq:error}
\end{align*}
where $\tilde{\mathbf{V}} = \mathbf{V} - \hat{\mathbf{V}}$. Given that $\hat{\mathbf{V}}$ is fixed, $\mathbf{\epsilon}_1$ and $\mathbf{\epsilon}_2$ are independent and the covariance matrix of the error in \eqref{eq:error} is given by
\begin{align*}
\mathbb{E}[(\mathbf{\epsilon}_1 + \mathbf{\epsilon}_2)(\mathbf{\epsilon}_1 + \mathbf{\epsilon}_2)^\dagger] &= \mathbb{E}[\mathbf{\epsilon}_1\mathbf{\epsilon}_1^\dagger] + \mathbb{E}_{\mathbf{n}}[\mathbf{\epsilon}_2\mathbf{\epsilon}_2^\dagger],
\numberthis\label{eq:errorvariance}
\end{align*}
where the expectation, $\mathbb{E}[\mathbf{\epsilon}_1\mathbf{\epsilon}_1^\dagger]$, is performed over the ray phases and angles in $\mathbf{h}^{\mathrm{\scriptscriptstyle RU}'}$, $\mathbf{h}^{\mathrm{\scriptscriptstyle BU}}$, and $\tilde{\mathbf{H}}^{\mathrm{\scriptscriptstyle BR}'}$. This does not apply to $\mathbb{E}[\mathbf{\epsilon}_2\mathbf{\epsilon}_2^\dagger]$, which involves none of these channel elements. Hence, the expectation is performed over the random noise at the BS.  Section~\ref{s:analysis} analyses each of the terms in \eqref{eq:errorvariance} in order to design efficient training phase matrices.

Using the channel estimates for the active RIS elements within $\hat{\mathbf{h}}^{\mathrm{\scriptscriptstyle RU}'}$, a complete estimate, $\hat{\mathbf{h}}^\mathrm{\scriptscriptstyle RU}$, is obtained using one of three interpolation methods. The one-point ($1$-pt) method, requiring $N' = N_y$, estimates the channel for the middle RIS element of each column and assumes all elements in the column have the same channel coefficient. This is similar to the standard grouping approach used in \cite{You2020Intelligent,Zheng2020Intelligent}. The $2$-pt method, requiring $N' = 2N_y$, uses channel estimates from the top and bottom RIS element in each column, interpolating linearly between them on the complex plane. Finally, the $3$-pt method, requiring $N' = 3N_y$, estimates the channel for the top, middle, and bottom RIS elements in each column. By centering the channel coefficient from the middle RIS element at the origin of the complex plane and rotating the other two estimated coefficients about this point, a family of quadratic curves can be fitted. The $3$-pt method interpolates equidistant channel coefficients on the parabola with the minimum curvature. 
For comparison, we additionally consider an $N_z$-pt method where all RIS elements are active during Stage 1, requiring $N' = N$ and no interpolation.

The known value of $\bar{\mathbf{H}}^\mathrm{\scriptscriptstyle BR}$ and the estimates of $\hat{\mathbf{h}}^\mathrm{\scriptscriptstyle RU}$ and $\hat{\mathbf{h}}^\mathrm{\scriptscriptstyle BU}$ are then used to configure the RIS phases to be used during transmission, $\boldsymbol{\bar{\Phi}}$, using the desired RIS phase design. The initial estimate of the $M\times 1$ total channel $\hat{\mathbf{h}}^\mathrm{\scriptscriptstyle TOT}_{\tau_1}$ derived from CSI obtained in Stage 1 is then:
\begin{align}
\hat{\mathbf{h}}^\mathrm{\scriptscriptstyle TOT}_{\tau_1} = \bar{\mathbf{H}}^\mathrm{\scriptscriptstyle BR}\boldsymbol{\bar{\Phi}}\hat{\mathbf{h}}^\mathrm{\scriptscriptstyle RU} + \hat{\mathbf{h}}^\mathrm{\scriptscriptstyle BU}.\label{eq:htot_tau_1}
\end{align}
%
\subsection{Stage 2: Refine CSI for Uplink Beamforming}\label{ss:stage_2}
Following the configuration of the RIS phases for data transmission, the procedure in \cite{miller2021efficient} uses additional pilot symbols to re-estimate the total channel in Stage 2. This provides a more accurate total channel estimate, $\hat{\mathbf{h}}^\mathrm{\scriptscriptstyle TOT}_{\tau}$, for UL processing during data transmission.

The UE transmits a further $\tau_2$ symbols which are received at the BS as
\begin{align}
\mathbf{y}_t = \sqrt{\rho}\mathbf{h}^\mathrm{\scriptscriptstyle TOT}s_t + \mathbf{n}_t,
\end{align}
with $\mathbf{n}_t\sim\mathcal{CN}(0,\mathbf{I}_M)$ for $t = \tau_1+1 \dots \tau_1+\tau_2$. Again, the pilot symbols, $s_t$, have no bearing on the estimation and we continue to assume that $s_t = 1$. This provides $\tau_2$ additional measurements of the total channel which are used to refine $\hat{\mathbf{h}}_{\tau_1}^\mathrm{\scriptscriptstyle TOT}$ as
\begin{align}
\hat{\mathbf{h}}^\mathrm{\scriptscriptstyle TOT}_{\tau} = \omega_1\hat{\mathbf{h}}^\mathrm{\scriptscriptstyle TOT}_{\tau_1} + \frac{\omega_2}{\tau_2\sqrt{\rho}}\sum_{t = \tau_1 + 1}^{\tau}\mathbf{y}_{t},
\end{align}
with $\tau = \tau_1+\tau_2$. We consider two weighting vectors. $[\omega_1, \omega_2] = [1, \tau_2]/(\tau_2 + 1)$ \textit{refines} $\hat{\mathbf{h}}^\mathrm{\scriptscriptstyle TOT}_{\tau}$ using a weighted average of the initial estimate from Stage 1, and the new channel measurements from Stage 2. Alternatively, $[\omega_1, \omega_2] = [0,1]$, \textit{re-estimates} $\hat{\mathbf{h}}^\mathrm{\scriptscriptstyle TOT}_{\tau}$ by averaging the measurements in Stage 2 only. In this case, Stage 1 is exclusively used to set the RIS phases for data transmission, and Stage 2 is used to estimate the resulting end-to-end total channel.
%
\section{Error Variance Analysis}\label{s:analysis}
We begin by examining the error variance in order to design  efficient training phases. Without loss of generality, we perform analysis for the $N_z$-pt method with no interpolation. These results can be applied to analyse the estimation error for the \textit{active} elements of the interpolating methods by setting $N = N'$, $\mathbf{h}^{\mathrm{\scriptscriptstyle RU}} = \mathbf{h}^{\mathrm{\scriptscriptstyle RU}'}$, $\bar{\mathbf{H}}^{\mathrm{\scriptscriptstyle BR}} = \bar{\mathbf{H}}^{\mathrm{\scriptscriptstyle BR}'}$, etc.
\subsection{Error Variance from Processed Noise}\label{ss:e2}
Consider the estimation error arising from the processed noise term, $\epsilon_2$, in \eqref{eq:error} and \eqref{eq:errorvariance}. We note that this additionally represents the total error variance under pure LoS conditions in the BS-RIS channel, when $\epsilon_1 = \mathbf{0}^{(N'+M)\times1}$.

From \eqref{eq:error} and \eqref{eq:Vhat}, we have
\begin{align*}
\mathbb{E}[\epsilon_2\epsilon_2^\dagger] &= (\hat{\mathbf{V}}^\dagger\hat{\mathbf{V}})^{-1}\\
&= \rho^{-1}\begin{bmatrix}
\mathbf{A} & \mathbf{C}^\dagger\\
\mathbf{C} & \mathbf{D}
\end{bmatrix}^{-1}\\
& \triangleq \rho^{-1}\mathbf{X},\numberthis\label{eq:X_definition}
\end{align*}
with
\begin{align*}
\begin{rcases}
\mathbf{A} &= \sum_{t=1}^{\tau_1}\boldsymbol{\Psi}_{t}^\dagger\bar{\mathbf{H}}^{\mathrm{\scriptscriptstyle BR}\dagger}\bar{\mathbf{H}}^{\mathrm{\scriptscriptstyle BR}}\boldsymbol{\Psi}_{t} = M\bar{\beta}^{\mathrm{\scriptscriptstyle BR}}\mathbf{A}^\mathrm{\scriptscriptstyle R}\boldsymbol{\Omega}\mathbf{A}^{\mathrm{\scriptscriptstyle R}\dagger}\\
\mathbf{C} &= \bar{\mathbf{H}}^{\mathrm{\scriptscriptstyle BR}}\sum_{t=1}^{\tau_1}\boldsymbol{\Psi}_{t} = \sqrt{\bar{\beta}^{\mathrm{\scriptscriptstyle BR}}}\mathbf{a}^{\mathrm{\scriptscriptstyle B}}(\bar{\phi}^{\mathrm{\scriptscriptstyle B}},\bar{\theta}^{\mathrm{\scriptscriptstyle B}})\boldsymbol{\omega}^\dagger\mathbf{A}^{\mathrm{\scriptscriptstyle R}\dagger}\\
\mathbf{D} &= \tau_1\mathbf{I}_\mathrm{M},
\end{rcases}\numberthis\label{eq:ABD}
\end{align*}
where $\boldsymbol{\Omega} = \sum_{t=1}^{\tau_1}\boldsymbol{\psi}^{*}_{t}\boldsymbol{\psi}^{T}_{t}$, $\boldsymbol{\omega} = \sum_{t=1}^{\tau_1}\boldsymbol{\psi}^{*}_{t}$, and $\mathbf{A}^{\mathrm{\scriptscriptstyle R}} = \mathrm{diag}(\mathbf{a}^{\mathrm{\scriptscriptstyle R}}(\bar{\phi}^{\mathrm{\scriptscriptstyle R}},\bar{\theta}^{\mathrm{\scriptscriptstyle R}}))$. Using well-known results on the inverse of partitioned matrices \cite[0.7.3.1]{horn_johnson}, we obtain
\begin{align*}
\mathbf{X} = \begin{bmatrix}
\mathbf{X}_{11} & \mathbf{X}_{21}^\dagger\\
\mathbf{X}_{21} & \mathbf{X}_{22}
\end{bmatrix},\numberthis\label{eq:X_matrix}
\end{align*}
with
\begin{align*}
\begin{rcases}
\mathbf{X}_{11} &= [\mathbf{A} - \mathbf{C}^\dagger\mathbf{D}^{-1}\mathbf{C}]^{-1}\\
\mathbf{X}_{21} &= [\mathbf{C}\mathbf{A}^{-1}\mathbf{C}^\dagger - \mathbf{D}]^{-1}\mathbf{C}\mathbf{A}^{-1}\\
\mathbf{X}_{22} &= [\mathbf{D} - \mathbf{C}\mathbf{A}^{-1}\mathbf{C}^\dagger]^{-1}.
\end{rcases}\numberthis\label{eq:X_parts_general}
\end{align*}

To obtain the estimation error variances, we need only analyse the diagonal elements of $\mathbf{X}_{11}$ and $\mathbf{X}_{22}$, with the former corresponding to the error in $\hat{\mathbf{h}}^{\mathrm{\scriptscriptstyle RU}}$ and the latter to the error in $\hat{\mathbf{h}}^{\mathrm{\scriptscriptstyle BU}}$. Noting that $\mathbf{C}$ is rank-1, $\mathbf{X}_{11}$, $\mathbf{X}_{21}$, and $\mathbf{X}_{22}$ all require the inverse of a matrix with a rank-1 adjustment. Hence, we can use the inversion formula in \cite[0.7.4.1]{horn_johnson}, and the fact that $\mathbf{A}^{\mathrm{\scriptscriptstyle R}}$ is a unitary matrix, to obtain
\begin{align*}
\mathbf{X}_{11} = (M\bar{\beta}^\mathrm{\scriptscriptstyle BR})^{-1}[&\mathbf{A}^\mathrm{\scriptscriptstyle R}\mathbf{\Omega}^{-1}\mathbf{A}^{\mathrm{\scriptscriptstyle R}\dagger}\\
&+\alpha^{-1}\mathbf{A}^\mathrm{\scriptscriptstyle R}\boldsymbol{\Omega}^{-1}\boldsymbol{\omega}\boldsymbol{\omega}^\dagger \boldsymbol{\Omega}^{-1}\mathbf{A}^{\mathrm{\scriptscriptstyle R}\dagger}],\numberthis\label{eq:X11_general}
\end{align*}
\begin{align*}
\mathbf{X}_{21} = -(\alpha^2 M^2\bar{\beta}^\mathrm{\scriptscriptstyle BR})^{-\frac{1}{2}}\mathbf{a}^{\mathrm{\scriptscriptstyle B}}\boldsymbol{\omega}^{\dagger}\boldsymbol{\Omega^{-1}}\mathbf{A}^{\mathrm{\scriptscriptstyle R}\dagger},\numberthis\label{eq:X21_general}
\end{align*}
and
\begin{align*}
\mathbf{X}_{22} = (\tau_1)^{-1}\Big[\mathbf{I}_M + \Big(\frac{\boldsymbol{\omega}^\dagger\boldsymbol{\Omega}^{-1}\boldsymbol{\omega}}{M\alpha}\Big)\mathbf{a}^{\mathrm{\scriptscriptstyle B}}(\bar{\phi}^{\mathrm{\scriptscriptstyle B}},\bar{\theta}^{\mathrm{\scriptscriptstyle B}})\mathbf{a}^{\mathrm{\scriptscriptstyle B}\dagger}(\bar{\phi}^{\mathrm{\scriptscriptstyle B}},\bar{\theta}^{\mathrm{\scriptscriptstyle B}})\Big],\numberthis\label{eq:X22_general}
\end{align*}
with $\alpha = \tau_1 - \boldsymbol{\omega}^\dagger\boldsymbol{\Omega}^{-1}\boldsymbol{\omega}$. These results are used in Sec.~\ref{s:design} to design the training phase matrix.
%
\subsection{Error Variance from Unknown Scattering in the RIS-BS Channel}\label{ss:e1}
We now simplify $\boldsymbol{\epsilon}_1$, the error introduced from unknown scattering in the RIS-BS channel. From \eqref{eq:error}, we have
\begin{align*}
\mathbf{\epsilon}_1 &= (\hat{\mathbf{V}}^\dagger\hat{\mathbf{V}})^{-1}\hat{\mathbf{V}}^\dagger\tilde{\mathbf{V}}\begin{bmatrix}
\mathbf{h}^{\mathrm{\scriptscriptstyle RU}}\\
\mathbf{h}^\mathrm{\scriptscriptstyle BU}
\end{bmatrix}\\
&\triangleq\mathbf{X}\begin{bmatrix}
\tilde{\mathbf{A}} & \tilde{\mathbf{B}}\\
\tilde{\mathbf{C}} & \tilde{\mathbf{D}}\\
\end{bmatrix}\begin{bmatrix}
\mathbf{h}^{\mathrm{\scriptscriptstyle RU}}\\
\mathbf{h}^\mathrm{\scriptscriptstyle BU}
\end{bmatrix},
\end{align*}
using \eqref{eq:X_definition} and giving $\rho^{-1}\hat{\mathbf{V}}^\dagger\tilde{\mathbf{V}}$ in block matrix form with
\begin{align*}
\begin{rcases}
\tilde{\mathbf{A}} &= \sum_{t = 1}^{\tau_1}\boldsymbol{\Psi}_{t}^{\dagger}\bar{\mathbf{H}}^{\mathrm{\scriptscriptstyle BR}\dagger}\tilde{\mathbf{H}}^{\mathrm{\scriptscriptstyle BR}}\boldsymbol{\Psi}_{t}\\
\tilde{\mathbf{B}} &= \mathbf{0}^{N\times M}\\
\tilde{\mathbf{C}} &= \tilde{\mathbf{H}}^{\mathrm{\scriptscriptstyle BR}}\mathrm{diag}(\boldsymbol{\omega}^*)\\
\tilde{\mathbf{D}} &= \mathbf{0}^{M\times M},
\end{rcases}\numberthis\label{eq:ABCDtilde}
\end{align*}
which follows from the definitions of $\hat{\mathbf{V}}$ and $\tilde{\mathbf{V}}$. Hence,
\begin{align*}
\epsilon_1&= \begin{bmatrix}
\mathbf{X}_{11}\tilde{\mathbf{A}}\mathbf{h}^{\mathrm{\scriptscriptstyle RU}} + \mathbf{X}_{21}^{\dagger}\tilde{\mathbf{C}}\mathbf{h}^{\mathrm{\scriptscriptstyle RU}}\\
\mathbf{X}_{21}\tilde{\mathbf{A}}\mathbf{h}^{\mathrm{\scriptscriptstyle RU}} + \mathbf{X}_{22}\tilde{\mathbf{C}}\mathbf{h}^{\mathrm{\scriptscriptstyle RU}}
\end{bmatrix}\\
&\triangleq\begin{bmatrix}
\mathbf{Y}_{11}\\
\mathbf{Y}_{22}\numberthis\label{eq:Ydef}
\end{bmatrix}.
\end{align*}
The diagonal elements of $\boldsymbol{\epsilon}_{1}\boldsymbol{\epsilon}_{1}^{\dagger}$ are then comprised of the diagonal elements of $\mathbf{Y}_{11}\mathbf{Y}_{11}^{\dagger}$, which give the $\epsilon_1$ component of the error variance in $\hat{\mathbf{h}}^{\mathrm{\scriptscriptstyle RU}}$, and the diagonal elements of $\mathbf{Y}_{22}\mathbf{Y}_{22}^{\dagger}$, which give the same for $\hat{\mathbf{h}}^{\mathrm{\scriptscriptstyle BU}}$.
Using the definitions of $\mathbf{X}_{11}$ and $\mathbf{X}_{21}$, we simplify the first term in $\mathbf{Y}_{11}$ as
\begin{align*}
\mathbf{X}_{11}\tilde{\mathbf{A}}\mathbf{h}^{\mathrm{\scriptscriptstyle RU}} &= (M\bar{\beta}^{\mathrm{\scriptscriptstyle BR}})^{-1}[\mathbf{A}^{\mathrm{\scriptscriptstyle R}}\boldsymbol{\Omega}^{-1}\mathbf{A}^{\mathrm{\scriptscriptstyle R}\dagger}\tilde{\mathbf{A}}\mathbf{h}^{\mathrm{\scriptscriptstyle RU}}\\
&+ \alpha^{-1}\mathbf{A}^{\mathrm{\scriptscriptstyle R}}\boldsymbol{\Omega}^{-1}\boldsymbol{\omega}\boldsymbol{\omega}^{\dagger}
\boldsymbol{\Omega}^{-1}\mathbf{A}^{\mathrm{\scriptscriptstyle R}\dagger}\tilde{\mathbf{A}}\mathbf{h}^{\mathrm{\scriptscriptstyle RU}}].\numberthis\label{eq:Y11_1}
\end{align*}
Now,
\begin{align*}
&\boldsymbol{\Omega}^{-1}\mathbf{A}^{\mathrm{\scriptscriptstyle R}\dagger}\tilde{\mathbf{A}}\mathbf{h}^{\mathrm{\scriptscriptstyle RU}}\\
&= \sqrt{\bar{\beta}^{\mathrm{\scriptscriptstyle BR}}}\boldsymbol{\Omega}^{-1}\sum_{t=1}^{\tau_1}\boldsymbol{\psi}_{t}^{*}\mathbf{a}^{\mathrm{\scriptscriptstyle B}\dagger}(\bar{\phi}^{\mathrm{\scriptscriptstyle B}}, {\bar{\theta}^\mathrm{\scriptscriptstyle B}})\tilde{\mathbf{H}}^{\mathrm{\scriptscriptstyle BR}}\mathrm{diag}(\mathbf{h}^{\mathrm{\scriptscriptstyle RU}})\boldsymbol{\psi}_{t}\\
&\triangleq\sqrt{\bar{\beta}^{\mathrm{\scriptscriptstyle BR}}} \boldsymbol{\Omega}^{-1}\sum_{t=1}^{\tau_1}\boldsymbol{\psi}_{t}^{*}\Lambda\boldsymbol{\psi}_{t}\\
&=\sqrt{\bar{\beta}^{\mathrm{\scriptscriptstyle BR}}}\boldsymbol{\Omega}^{-1}\sum_{t=1}^{\tau_1}\boldsymbol{\psi}_{t}^{*}\boldsymbol{\psi}_{t}^{\mathrm{T}}\Lambda^{\mathrm{T}}\\
&=\sqrt{\bar{\beta}^{\mathrm{\scriptscriptstyle BR}}}\Lambda^{\mathrm{T}},\numberthis\label{eq:OmegaAAhru}
\end{align*}
where the $1\times N$ vector $\Lambda$ is independent of training phases. Substituting \eqref{eq:OmegaAAhru} into \eqref{eq:Y11_1}, we obtain
\begin{align*}
\mathbf{X}_{11}\tilde{\mathbf{A}}\mathbf{h}^{\mathrm{\scriptscriptstyle RU}} &= (M^2\bar{\beta}^{\mathrm{\scriptscriptstyle BR}})^{-\frac{1}{2}}(\mathbf{A}^{\mathrm{\scriptscriptstyle R}} + \alpha^{-1}\mathbf{A}^{\mathrm{\scriptscriptstyle R}}\boldsymbol{\Omega}^{-1}\boldsymbol{\omega}\boldsymbol{\omega}^{\dagger})\Lambda^{\mathrm{T}}.
\end{align*}
Following a similar process for the second term of $\mathbf{Y}_{11}$ yields
\begin{align*}
\mathbf{X}_{21}^{\dagger}\tilde{\mathbf{C}}\mathbf{h}^{\mathrm{\scriptscriptstyle RU}} &= -(\alpha^2 M^2\bar{\beta}^{\mathrm{\scriptscriptstyle BR}})^{-\frac{1}{2}}\mathbf{A}^{\mathrm{\scriptscriptstyle R}}\boldsymbol{\Omega}^{-1}\boldsymbol{\omega}\boldsymbol{\omega}^{\dagger}\Lambda^{\mathrm{T}}.\numberthis\label{eq:Y11_2}
\end{align*}
Hence, $\mathbf{Y}_{11}$ becomes 
\begin{align*}
\mathbf{Y}_{11} &= \big(M^2\bar{\beta}^{\mathrm{\scriptscriptstyle BR}}\big)^{-\frac{1}{2}}\mathbf{A}^{\mathrm{\scriptscriptstyle R}}\Lambda^{\mathrm{T}}.\numberthis\label{eq:Y11_final}
\end{align*}
The result in \eqref{eq:Y11_final} indicates that the error in $\hat{\mathbf{h}}^{\mathrm{\scriptscriptstyle RU}}$ caused by unknown scattering in the RIS-BS channel cannot be affected by the training phase design as $\mathbf{A}$ and $\Lambda$ are independent of the training phases.

Analysing the first term in $\mathbf{Y}_{22}$ and using the result in \eqref{eq:OmegaAAhru}, we have
\begin{align*}
\mathbf{X}_{21}\tilde{\mathbf{A}}\mathbf{h}^{\mathrm{\scriptscriptstyle RU}} &= -(\alpha M)^{-1}\mathbf{a}^{\mathrm{\scriptscriptstyle B}}(\bar{\phi}^{\mathrm{\scriptscriptstyle B}}, \bar{\theta}^{\mathrm{\scriptscriptstyle B}})\Lambda\boldsymbol{\omega}^{*},\numberthis\label{eq:Y22_1}
\end{align*}
while the second term becomes
\begin{align*}
\mathbf{X}_{22}\tilde{\mathbf{C}}\mathbf{h}^{\mathrm{\scriptscriptstyle RU}} &= \tau_1^{-1}[\tilde{\mathbf{H}}^{\mathrm{\scriptscriptstyle BR}}\mathrm{diag}(\mathbf{h}^{\mathrm{\scriptscriptstyle RU}})\\
&+ (M\alpha)^{-1}\boldsymbol{\omega}^{\dagger}\boldsymbol{\Omega}^{-1}\boldsymbol{\omega}\mathbf{a}^{\mathrm{\scriptscriptstyle B}}(\bar{\phi}^{\mathrm{\scriptscriptstyle B}}, \bar{\theta}^{\mathrm{\scriptscriptstyle B}})\Lambda]\boldsymbol{\omega}^{*}.\numberthis\label{eq:Y22_2}
\end{align*}
Using the definition of $\alpha$ and performing some straightforward algebra gives
\begin{align*}
\mathbf{Y}_{22} = (\tau_1M)^{-1}[\tilde{M\mathbf{H}}^{\mathrm{\scriptscriptstyle BR}}\mathrm{diag}(\mathbf{h}^{\mathrm{\scriptscriptstyle RU}}) - \mathbf{a}^{\mathrm{\scriptscriptstyle B}}(\bar{\phi}^{\mathrm{\scriptscriptstyle B}}, \bar{\theta}^{\mathrm{\scriptscriptstyle B}})\Lambda]\boldsymbol{\omega}^{*}.\numberthis\label{eq:Y22_3}
\end{align*}
The analysis in Sec.~\ref{ss:e2} and~\ref{ss:e1} is used to inform the training phase design in Sec.~\ref{s:design}.
%
%
\section{Optimal training phase design}\label{s:design}
Having simplified the terms in \eqref{eq:errorvariance}, we find that the error variance is dictated by $\mathbf{X}_{11}$, $\mathbf{X}_{22}$, $\mathbf{Y}_{22}\mathbf{Y}_{22}^{\dagger}$, and $\mathbf{Y}_{11}\mathbf{Y}_{11}^{\dagger}$, the latter of which is independent of training phases. Hence, this section presents the MDFT training matrix which  minimizes the trace of the remaining three variables, and therefore the sum of the error variances.

The MDFT design consists of an $N\times \tau_1$ \textit{modified DFT} matrix of training phases, constructed from the last $N$ rows of a $\tau_1\times \tau_1$ DFT matrix with $\tau_1 = N+1$, such that
	\begin{align*}
	\boldsymbol{\Phi}^{\mathrm{\scriptscriptstyle MD}} =
	\begin{bmatrix}
	1 & e^{jx} & e^{j2x} & e^{j3x} & \cdots & e^{jNx}\\
	1 & e^{j2x} & e^{j4x} & e^{j6x} & \cdots & e^{j2Nx}\\
	1 & e^{j3x} & e^{j6x} & e^{j9x} & \cdots & e^{j3Nx}\\
	\vdots & \vdots & \vdots & \vdots & \ddots & \vdots\\
	1 & e^{jNx} & e^{j2Nx} & e^{j3Nx} & \cdots & e^{j{N^2}x}
	\end{bmatrix}
	\end{align*}
	with $x = 2\pi/(N+1)$.
\begin{result}\label{res:phases}
The MDFT design is optimal in the sense that it minimizes the sum of estimation error variances given by  $\mathrm{tr}(\mathbf{X}_{11}+\mathbf{X}_{22}+\mathbf{Y}_{11}\mathbf{Y}_{11}^{\dagger}+\mathbf{Y}_{22}\mathbf{Y}_{22}^{\dagger})$.
\end{result}
\begin{proof}
See Appendix~\ref{App}.
\end{proof}
Next, we compare the MDFT design  with the DFT approach \cite{Wan2020Broadband,Jensen2020Optimal,Nadeem2020Intelligent,Zheng2020Intelligent} using $\mathrm{tr}( \mathbf{X})$ for simplicity as it is the dominant error covariance matrix. Using the MDFT approach we have
\begin{align*}
\boldsymbol{\Omega}^{\mathrm{\scriptscriptstyle MD}} &= (N+1)\mathbf{I}_{N},\numberthis\label{eq:Ow}\\
\boldsymbol{\omega}^{\mathrm{\scriptscriptstyle MD}} &= \mathbf{0}^{N\times 1}.
\end{align*}
Substituting \eqref{eq:Ow} into \eqref{eq:ABD} gives
\begin{align*}
\mathbf{A}^{\mathrm{\scriptscriptstyle MD}} &= M\bar{\beta}^{\mathrm{\scriptscriptstyle BR}}(N+1)\mathbf{I}_N\\
\mathbf{C}^{\mathrm{\scriptscriptstyle MD}} &= \mathbf{0}^{M\times N}.\numberthis\label{eq:AC_moddft}
\end{align*}
%
This reduces $\mathrm{tr}(\mathbf{X}_{11})$ and $\mathrm{tr}(\mathbf{X}_{22})$ to
\begin{align*}
\mathrm{tr}(\mathbf{X}_{11}^{\mathrm{\scriptscriptstyle MD}}) = N(M\bar{\beta}^{\mathrm{\scriptscriptstyle BR}}(N+1))^{-1}
\end{align*}
and
\begin{align*}
\mathrm{tr}(\mathbf{X}_{22}^{\mathrm{\scriptscriptstyle MD}}) = M(N+1)^{-1}.
\end{align*}
For the DFT approach, we follow the general procedure in \cite{miller2021efficient} where $\tau_1 = N+1$ training vectors are used with the first $N$ being the columns of an $N\times N$ DFT matrix, $\mathbf{W}$. The $\tau_1^\mathrm{th}$ vector is allowed to be an arbitrary training phase vector, $\mathbf{w}=[w_1,w_2, \ldots , w_N]^T$. With this notation, we obtain
\begin{align*}
\boldsymbol{\Omega}^{\mathrm{\scriptscriptstyle DFT}} &= (\mathbf{W}\mathbf{W}^\dagger + \mathbf{w}\mathbf{w}^\dagger)^* = N\mathbf{I}_N + \mathbf{w}^*\mathbf{w}^\mathrm{T}\\
\boldsymbol{\omega}^{\mathrm{\scriptscriptstyle DFT}} &= N\mathbf{1}^{N\times 1}_{11} + \mathbf{w}^*,
\end{align*}
where $\mathbf{1}^{a\times b}_{c,d}$ is an $a\times b$ matrix where the element in row $c$ and column $d$ is 1 and all others are 0. With these results and using \cite[0.7.4.2]{horn_johnson}, we have
\begin{align*}
\mathrm{tr}(\mathbf{X}_{11}^{\mathrm{\scriptscriptstyle DFT}}) &= (M\bar{\beta}^{\mathrm{\scriptscriptstyle BR}})^{-1}(1 -1/N - (\mathrm{Re}({w}_1)-1)^{-1}),\numberthis\label{eq:trace_X11_DFT}
\end{align*}
which depends only on the first element of $\mathbf{w}$, assuming the remaining entries of $\mathbf{w}$ are non-zero, and is minimized at ${w}_1 = e^{j\pi} = -1$, giving
\begin{align}
\mathrm{tr}(\mathbf{X}_{11}^{\mathrm{\scriptscriptstyle DFT}}) & \ge (M\bar{\beta}^{\mathrm{\scriptscriptstyle BR}})^{-1}\left(\frac{3}{2}-\frac{1}{N}\right)\notag\\
&=\mathrm{tr}(\mathbf{X}_{11}^{\mathrm{\scriptscriptstyle MD}}) + (M\bar{\beta}^{\mathrm{\scriptscriptstyle BR}})^{-1}\left(\frac{1}{2} - (N(N+1))^{-1}\right),\label{firstDFT}\\
\mathrm{tr}(\mathbf{X}_{22}^{\mathrm{\scriptscriptstyle DFT}}) & \ge (\tau_1^2)^{-1}\left(\tau_1 M + \frac{1}{2}(N^2-1)\right)\notag\\ &=\mathrm{tr}(\mathbf{X}_{22}^{\mathrm{\scriptscriptstyle MD}}) + \frac{1}{2} - (N+1)^{-1}.\label{secondDFT}
\end{align}
From Theorem 1, the MDFT approach must outperform the best DFT approach. This is verified by \eqref{firstDFT}-\eqref{secondDFT} but the improvements are seen to be small.

\section{Simplified SNR analysis}\label{simplified}
While simulations can give precise details of the performance of the channel estimation schemes, it is also useful to obtain some analytical insight into the effect of channel estimation error on overall performance. For performance, we consider the SNR of a single user using the optimal RIS design in \cite{Singh2020Optimal}.  Here, the optimal RIS transmission phases are given by:
\begin{align*}
\boldsymbol{{\Phi}}^*&= \nu \mathrm{diag}(\mathbf{a}^\mathrm{\scriptscriptstyle R}(\bar{\phi}^\mathrm{\scriptscriptstyle R},\bar{\theta}^\mathrm{\scriptscriptstyle R}))\mathrm{diag}(\exp(j\angle{\mathbf{h}}^\mathrm{RU\dagger})),\numberthis\label{eq:optimal_phases}
\end{align*}
where $\angle{\mathbf{h}}^\mathrm{\scriptscriptstyle RU}$ is defined by $(\angle{\mathbf{h}}^\mathrm{\scriptscriptstyle RU})_r=({\mathbf{h}}^\mathrm{\scriptscriptstyle RU})_r/|({\mathbf{h}}^\mathrm{\scriptscriptstyle RU})_r|$, and
\begin{align*}
\nu = \frac{\mathbf{a}^{\mathrm{\scriptscriptstyle B}\dagger}(\bar{\phi}^\mathrm{\scriptscriptstyle B},\bar{\theta}^\mathrm{\scriptscriptstyle B}){\mathbf{h}}^\mathrm{\scriptscriptstyle BU}}{|\mathbf{a}^{\mathrm{\scriptscriptstyle B}\dagger}(\bar{\phi}^\mathrm{\scriptscriptstyle B},\bar{\theta}^\mathrm{\scriptscriptstyle B}){\mathbf{h}}^\mathrm{\scriptscriptstyle BU}|}.
\end{align*}

For the errors, in order to make analytical progress, we make the following simplifying assumptions. First, we assume the RIS phases are set based on errored CSI, but the total channel (using the imperfect RIS) is then perfectly known. Secondly, for the errors in the RIS, we use a simple AWGN model where the phases used are the optimal phases plus iid Gaussian error. Hence, $\bar{\boldsymbol{\Phi}} = \boldsymbol{\Phi}^{*}\times \mathrm{diag}(\mathbf{e})$ with $\mathbf{e}=[e_1, e_2, \ldots ,e_N]^T$ and $e_i\sim\mathcal{N}(0,\sigma_e^2)$.

 Using results from \cite{Singh2020Optimal} and the notation $ \mathbf{h}^\mathrm{\scriptscriptstyle RU}=[h_1^\mathrm{\scriptscriptstyle RU}, h_2^\mathrm{\scriptscriptstyle RU}, \ldots ,h_N^\mathrm{\scriptscriptstyle RU}]^T$, the resulting mean SNR, assuming pure LOS for the RIS-BS channel, can be written as
\begin{align}\label{meansnr}
&\mathbb{E}[\overline{\mathrm{SNR}}] = \rho^\mathrm{\scriptscriptstyle D}\mathbb{E}[\lVert\mathbf{h}^\mathrm{\scriptscriptstyle TOT}\rVert^2]\notag\\
&= \rho^\mathrm{\scriptscriptstyle D}\Bigg\{\hspace{-1mm}M\beta^\mathrm{\scriptscriptstyle BU}\hspace{-1mm}+\hspace{-1mm}2\sqrt{\beta^\mathrm{\scriptscriptstyle BR}}
\Xi
%
\mathbb{E}\Bigg[\sum_{n=1}^{N}|{h}_n^\mathrm{\scriptscriptstyle RU}|\Bigg]\mathbb{E}[\exp(je_i)]  + \beta^\mathrm{\scriptscriptstyle BR}M \notag\\
&\times \mathbb{E}\Bigg[\sum_{n=1}^{N}|{h}_n^\mathrm{\scriptscriptstyle RU}|^2+\sum_{n=1}^{N}\sum_{\substack{n'=1\\ n'\neq n}}^{N}|{h}_n^\mathrm{\scriptscriptstyle RU}||{h}_{n'}^\mathrm{\scriptscriptstyle RU}|\exp(j(e_n-e_{n'}))\Bigg]\Bigg\}\notag\\
&\geq \rho^\mathrm{\scriptscriptstyle D}\Bigg\{\hspace{-1mm}M\beta^\mathrm{\scriptscriptstyle BU}\hspace{-1mm}+\hspace{-1mm}2\sqrt{\beta^\mathrm{\scriptscriptstyle BR}}
\Xi
\mathbb{E}\Bigg[\sum_{n=1}^{N}|{h}_n^\mathrm{\scriptscriptstyle RU}|\Bigg]\mathbb{E}[\exp(je_n)]\notag\\
&+ \beta^\mathrm{\scriptscriptstyle BR}M\mathbb{E}\Bigg[\left(\sum_{n=1}^{N}|\mathbf{h}_n^\mathrm{\scriptscriptstyle RU}|\right)^2\Bigg]\mathbb{E}[\exp(j(e_n-e_{n'}))]\Bigg\}\notag\\
&\triangleq\rho^\mathrm{\scriptscriptstyle D}\{\mathrm{S}_1 + \mathrm{S}_2\exp(-\sigma_e^2/2) + \mathrm{S}_3\exp(-\sigma_e^2)\},
\end{align}
with $\Xi\triangleq\mathbb{E}\Big[|\mathbf{a}^{\mathrm{\scriptscriptstyle B}\dagger}(\bar{\phi}^\mathrm{\scriptscriptstyle B},\bar{\theta}^\mathrm{\scriptscriptstyle B})\mathbf{h}^\mathrm{\scriptscriptstyle BU}|\Big]$ and %
where the final step results from the properties of the Gaussian distributed phase errors. In contrast, the optimal value is $\mathbb{E}[{\mathrm{SNR}^*}]=\rho^\mathrm{\scriptscriptstyle D}\{\mathrm{S}_1 + \mathrm{S}_2+ \mathrm{S}_3\}$.

In the mean SNR, the $\mathrm{S}_3$ term dominates when the RIS-assisted path has reasonable power and it is here that the RIS design becomes most important as this is the situation where the RIS can provide substantial benefits. In this scenario, the SNR is reduced, approximately, by a scaling factor of $\exp(-\sigma_e^2)$. The corresponding drop in SE is
\begin{align}
\mathbb{E}[\mathrm{SE}^{*}]-\mathbb{E}[\overline{\mathrm{SE}}]&\approx\log_{2}\Big(\exp(\sigma_e^2)\Big)=\sigma_e^2\log_2(e).\numberthis\label{eq:SE_drop}
\end{align}
Although \eqref{eq:SE_drop} is extremely simple, it is surprisingly accurate and is used to explain the numerical results in Sec.~\ref{s:results}.
\section{Numerical Results}\label{s:results}
We consider a cell with radius $r = 100$m and exclusion zone $r_0 = 15$m. A VURA RIS with $N_y = 12$ columns and a BS with $M=32$ and $d^\mathrm{\scriptscriptstyle B} = 0.5$ are positioned at $-5$m and $5$m on the $y$-axis, respectively. For $\beta^\mathrm{\scriptscriptstyle BU}$ and $\beta^\mathrm{\scriptscriptstyle RU}$ we use $\Gamma = 3.7$ and $\sigma_\mathrm{sf} = 5.5\mathrm{dB}$ as in \cite{sangodoyin_cluster_2018}. This level of shadowing is included in all figures except Figs.~\ref{fig:NMSE_vs_K_narrow} and~\ref{fig:NMSE_vs_K_wide}, as discussed later.     For $\beta^\mathrm{\scriptscriptstyle BR}$ we set $\Gamma = 2$ and there is no shadowing.  From \cite{3GPP} we set $\eta = 0.1$ for all channels. During training, we set $\rho$ based on the pathloss through the RIS-assisted link such that $\rho/(\beta^\mathrm{\scriptscriptstyle RU}\beta^\mathrm{\scriptscriptstyle BR}) = 5\mathrm{dB}$.

In this section, the CE procedure performance is evaluated in three ways. We examine the estimation error in $\mathbf{h}^\mathrm{\scriptscriptstyle TOT}_{\tau_1}$ in \eqref{eq:htot_tau_1} after Stage 1, the resulting phase error in the RIS transmission phases, and the overall performance following re-estimation in Stage 2. To this end, we set the RIS transmission phases by substituting the estimated channels into \eqref{eq:optimal_phases}. This gives 
\begin{align*}
\boldsymbol{\bar{\Phi}}&= \hat{\nu} \mathrm{diag}(\mathbf{a}^\mathrm{\scriptscriptstyle R}(\bar{\phi}^\mathrm{\scriptscriptstyle R},\bar{\theta}^\mathrm{\scriptscriptstyle R}))\mathrm{diag}(\exp(j\angle\hat{\mathbf{h}}^\mathrm{RU\dagger})),\numberthis\label{eq:optimal_phases2}
\end{align*}
where $\angle\hat{\mathbf{h}}^\mathrm{\scriptscriptstyle RU}$ contains the phases of the elements in $\hat{\mathbf{h}}^\mathrm{\scriptscriptstyle RU}$, and
\begin{align*}
\hat{\nu} = \frac{\mathbf{a}^{\mathrm{\scriptscriptstyle B}\dagger}(\bar{\phi}^\mathrm{\scriptscriptstyle B},\bar{\theta}^\mathrm{\scriptscriptstyle B})\hat{\mathbf{h}}^\mathrm{\scriptscriptstyle BU}}{|\mathbf{a}^{\mathrm{\scriptscriptstyle B}\dagger}(\bar{\phi}^\mathrm{\scriptscriptstyle B},\bar{\theta}^\mathrm{\scriptscriptstyle B})\hat{\mathbf{h}}^\mathrm{\scriptscriptstyle BU}|}.
\end{align*}

Following re-estimation in Stage 2, $\hat{\mathbf{h}}^\mathrm{\scriptscriptstyle TOT}_{\tau}$ is used to perform MRC at the BS. The SNR for this system is given as
\begin{align}
\mathrm{SNR} = \frac{\rho^\mathrm{D}|\hat{\mathbf{h}}^{\mathrm{\scriptscriptstyle TOT}\dagger}_{\tau}\hat{\mathbf{h}}^{\mathrm{\scriptscriptstyle TOT}}_{\tau}|^2}{\rho^\mathrm{D}|\hat{\mathbf{h}}^{\mathrm{\scriptscriptstyle TOT}\dagger}_{\tau}(\hat{\mathbf{h}}^{\mathrm{\scriptscriptstyle TOT}}_{\tau}-\mathbf{h}^\mathrm{\scriptscriptstyle TOT})|^2 + \lVert\hat{\mathbf{h}}^\mathrm{\scriptscriptstyle TOT}_{\tau}\rVert^2},
\end{align}
where $\rho^\mathrm{D}$ is the uplink SNR during data transmission. The resulting SE, scaled by the proportion of data symbols, is
\begin{align}\label{eq:SE}
\mathrm{SE} = \frac{T-\tau}{T}\log_{2}(1 + \mathrm{SNR}).
\end{align}
%
We set $\rho^\mathrm{\scriptscriptstyle D}$ such that the median received SNR of the direct link, $\rho^{D}\beta^\mathrm{\scriptscriptstyle BU}$, is $0$ dB in the absence of shadowing.

Cluster central angles in azimuth are modelled as $\phi_c\sim\mathcal{CN}(0,\sigma_{\phi}^2)$ and in elevation  $\theta_c$ is Laplacian with a mean of $\pi/2$ and scale parameter $\sigma_{\theta}$. These settings are used for the two ray-based channels, $\mathbf{h}^\mathrm{\scriptscriptstyle RU}$ and $\mathbf{h}^\mathrm{\scriptscriptstyle BU}$. The mean cluster central angles for $\tilde{\mathbf{H}}^\mathrm{\scriptscriptstyle BR}$ are aligned with the corresponding LoS angles. Hence, $\phi_c^\mathrm{\scriptscriptstyle R}\sim\mathcal{CN}(\bar{\phi}^\mathrm{\scriptscriptstyle R},\sigma_{\phi}^2)$ and $\theta_c^\mathrm{\scriptscriptstyle R}$ is Laplacian with mean $\bar{\theta}^\mathrm{\scriptscriptstyle R}$ and scale parameter $\sigma_{\phi}$, and similarly for $\phi_c^\mathrm{\scriptscriptstyle B}$ and $\theta_c^\mathrm{\scriptscriptstyle B}$. All subray angular offsets are modelled using a Laplacian distribution with scale parameters $\sigma_{\Delta}$ in azimuth and $\sigma_{\delta}$ in elevation. These distributions are characterised using two sets of values, as shown in Table \ref{parameters}. Scenario 1 uses values from measurements at 2.53GHz \cite{sangodoyin_cluster_2018}, providing a narrow spread of ray angles, while Scenario 2 uses values from \cite{3GPP} for a wider angular spread. We use $C=3$ and $S=5$ for $\mathbf{h}^\mathrm{\scriptscriptstyle RU}$ and $\mathbf{h}^\mathrm{\scriptscriptstyle BU}$, and $C=2$ and $S=2$ for $\tilde{\mathbf{H}}^\mathrm{\scriptscriptstyle BR}$.

\begin{table}
	\renewcommand{\arraystretch}{1.5}
	\centering
	\caption{Model Parameters}
	\label{parameters}
	\begin{tabular}{|c|c|c|c|c|c| }
		\hline
		\textbf{Var.} & $\sigma_{\phi}$, $\sigma_{\Delta}$&$\sigma_{\theta}$, $\sigma_{\delta}$&$N_z$&$d^\mathrm{\scriptscriptstyle R}$&$K$\\
		\hline
		\textbf{Scen. 1}& $14.4^\mathrm{o}$, $1.9^\mathrm{o}$&$6.24^\mathrm{o}$, $1.37^\mathrm{o}$&$8$&$0.25$&$10^3$dB\\
		\hline
		\textbf{Scen. 2}&$31.64^\mathrm{o}$, $6.12^\mathrm{o}$&$24.25^\mathrm{o}$, $1.84^\mathrm{o}$&$16$&$0.5$&$12$dB\\
		\hline
	\end{tabular}
\end{table}

Figs.~\ref{fig:NMSE_vs_K_narrow} through~\ref{fig:NMSE_vs_dRIS_NLoS} examine the NMSE of the total channel estimate after Stage 1, $\hat{\mathbf{h}}_{\tau_1}^\mathrm{\scriptscriptstyle TOT}$ in \eqref{eq:htot_tau_1}. In order to explain these results it is useful to state the following two results from \cite{miller2021efficient}. Using the phase design in \eqref{eq:optimal_phases2}, Results 2 and 3 in \cite{miller2021efficient} provide simplified approximations of the NMSE for high and moderate values of $K$. Result 2 approximates the NMSE for $K\to\infty$ as
\begin{align*}
{\textrm{NMSE}}=\frac{\mathbb{E}[\lVert\mathbf{h}^\mathrm{\scriptscriptstyle TOT}-\hat{\mathbf{h}}_{\tau_1}^\mathrm{\scriptscriptstyle TOT}\rVert^2]}{\mathbb{E}\lVert\mathbf{h}^\mathrm{\scriptscriptstyle TOT}\rVert^2}\approx M\beta^\mathrm{\scriptscriptstyle BR}{A}_1,
\end{align*}
where
\begin{align}\label{A1}
{A}_1 = \frac{\mathbb{E}[\lvert\exp(j\angle\hat{\mathbf{h}}^{\mathrm{\scriptscriptstyle RU}\dagger})\boldsymbol{\epsilon}^\mathrm{\scriptscriptstyle RU}\rvert^2]}{\mathbb{E}[\lVert\mathbf{h}^\mathrm{\scriptscriptstyle TOT}\rVert^2]}
\end{align}
and $\boldsymbol{\epsilon}^\mathrm{\scriptscriptstyle RU} = \mathbf{h}^\mathrm{\scriptscriptstyle RU} - \hat{\mathbf{h}}^\mathrm{\scriptscriptstyle RU}$. For moderate $K$, Result 3 in \cite{miller2021efficient} gives
\begin{align*}
\frac{\mathbb{E}[|\mathbf{h}^\mathrm{\scriptscriptstyle TOT}-\hat{\mathbf{h}}_{\tau_1}^\mathrm{\scriptscriptstyle TOT}|^2]}{\mathbb{E}\lVert\mathbf{h}^\mathrm{\scriptscriptstyle TOT}\rVert^2}\approx\frac{\mathbb{E}[{A}_2]}{\mathbb{E}[\lVert\mathbf{h}^\mathrm{\scriptscriptstyle TOT}\rVert^2]},
\end{align*}
where
\begin{align*}
{A}_{2} &= M\sum\limits_{c=1}^{C}\sum\limits_{s=1}^{S}\beta^\mathrm{\scriptscriptstyle BR}_{c,s}\lvert\mathbf{a}^{\mathrm{\scriptscriptstyle R}\dagger}(\phi^\mathrm{\scriptscriptstyle R}_{c,s},\theta^\mathrm{\scriptscriptstyle R}_{c,s})\mathrm{diag}(\mathbf{a}^{\mathrm{\scriptscriptstyle R}}(\bar{\phi}^\mathrm{\scriptscriptstyle R},\bar{\theta}^\mathrm{\scriptscriptstyle R}))|\hat{\mathbf{h}}^\mathrm{\scriptscriptstyle RU}|\rvert^2,\numberthis\label{eq:A2}
\end{align*}
and $|\hat{\mathbf{h}}^\mathrm{\scriptscriptstyle RU}|$ denotes the vector of absolute values.

Figs.~\ref{fig:NMSE_vs_K_narrow} and~\ref{fig:NMSE_vs_K_wide} plot the NMSE of $\hat{\mathbf{h}}_{\tau_1}^\mathrm{\scriptscriptstyle TOT}$ with $N_z=8$ and $d^\mathrm{\scriptscriptstyle R} = 0.5$ for Scenarios 1 and 2, respectively, with random, DFT and MDFT training. Here, we neglect shadowing and only plot results for the random training phases in two cases. This is to avoid the very high variability caused by both random training phases and shadowing which necessitates extremely long simulation times. 

\begin{figure}[ht]
	\centering
	\includegraphics[trim=1cm 0.04cm 1.6cm 0.1cm, clip=true, width=1\columnwidth]{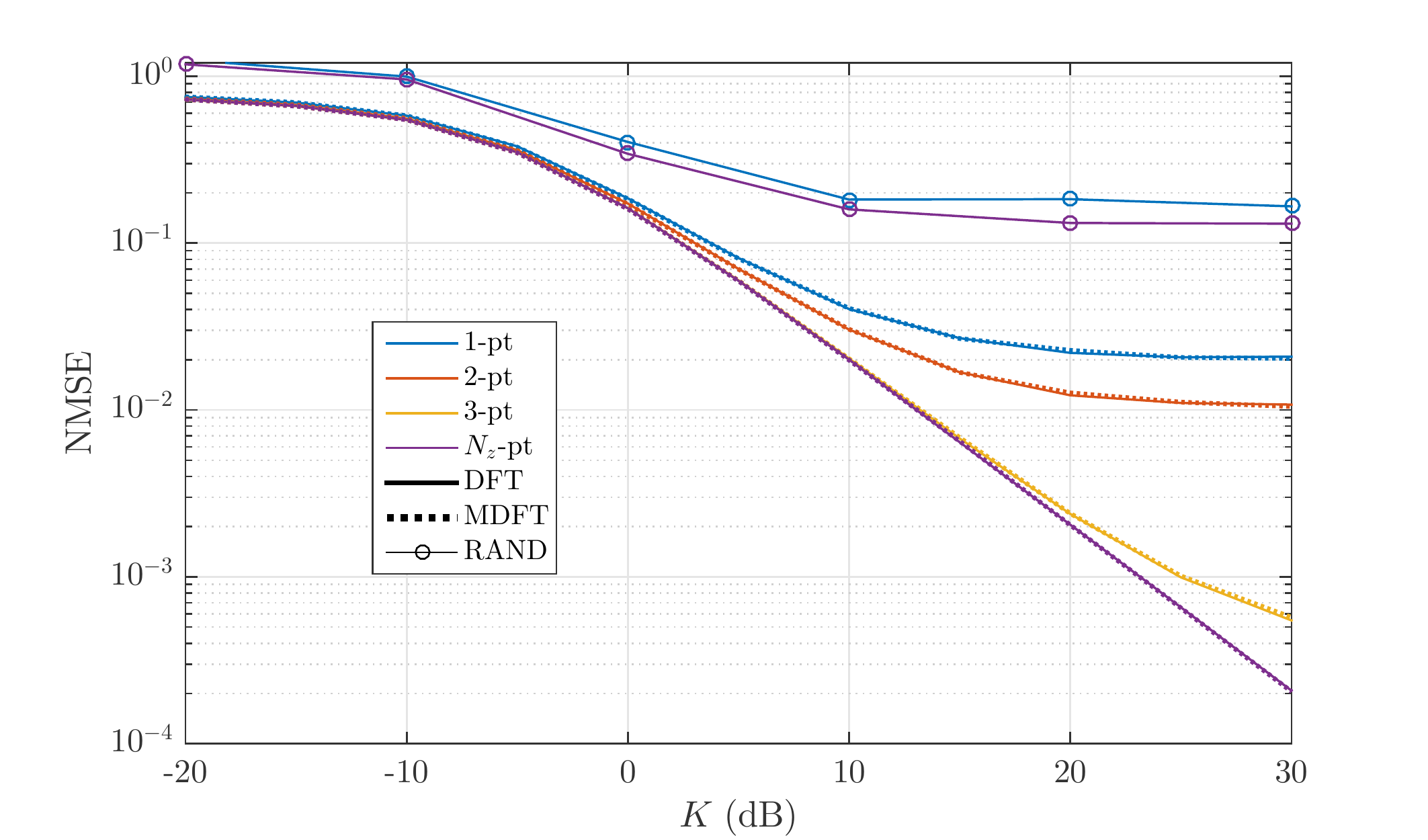}
	\caption{NMSE of $\hat{\mathbf{h}}_{\tau_1}^\mathrm{\scriptscriptstyle TOT}$ vs $K$ with narrow angular spread (Scenario 1 parameters are used except for the K-factor value which is varied).}
	\label{fig:NMSE_vs_K_narrow}
\end{figure}

\begin{figure}[ht]
	\centering
	\includegraphics[trim=1cm 0.04cm 1.6cm 0.1cm, clip=true, width=1\columnwidth]{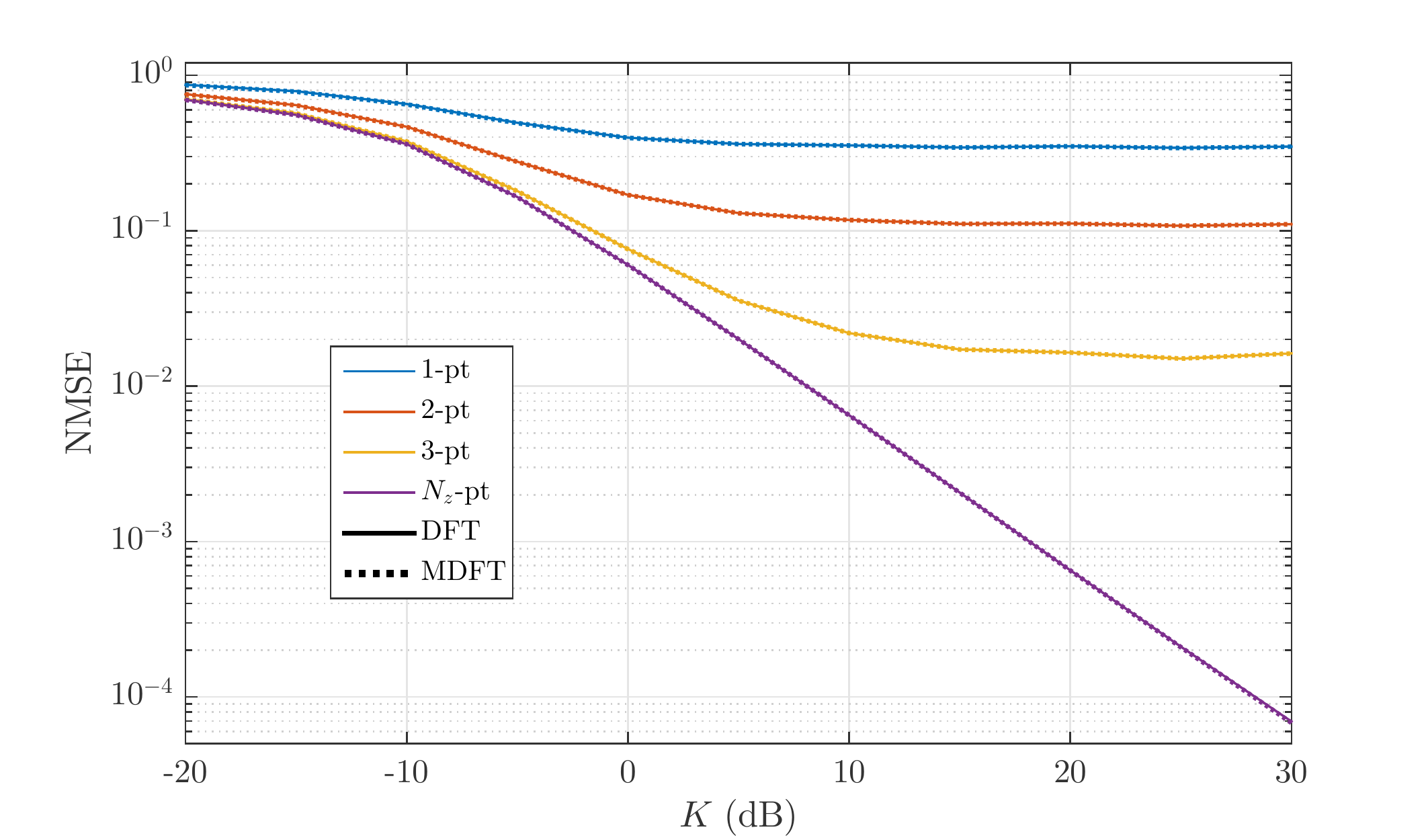}
	\caption{NMSE of $\hat{\mathbf{h}}_{\tau_1}^\mathrm{\scriptscriptstyle TOT}$ vs $K$ with wide angular spread (Scenario 2 parameters are used except for the K-factor value which is varied).}
	\label{fig:NMSE_vs_K_wide}
\end{figure}
Comparing the $N_z$-pt results in Figs.~\ref{fig:NMSE_vs_K_narrow} and~\ref{fig:NMSE_vs_K_wide}, which reflect the accuracy of channel estimation without the impact of interpolation, we observe a higher NMSE for a narrow angular spread than for wide angular spread. A narrower spread of ray angles in $\tilde{\mathbf{H}}^\mathrm{\scriptscriptstyle BR}$ enhances the alignment in \eqref{eq:A2} between the RIS LoS vector
$\mathbf{a}^\mathrm{\scriptscriptstyle R}(\bar{\phi}^\mathrm{\scriptscriptstyle R},\bar{\theta}^\mathrm{\scriptscriptstyle R})$  and the scattered rays $\mathbf{a}^\mathrm{\scriptscriptstyle R}(\phi^\mathrm{\scriptscriptstyle R}_{c,s},\theta^\mathrm{\scriptscriptstyle R}_{c,s})$ in $\tilde{\mathbf{H}}^\mathrm{\scriptscriptstyle BR}$, which increases the inner product with $|\hat{\mathbf{h}}^\mathrm{\scriptscriptstyle RU}|$, 
thus inflating ${A}_2$. %
We also note that NMSE increases with the level of interpolation and decreases with $K$. 
Furthermore, narrower angular spreads aid interpolation by increasing the correlation between channel coefficients in $\mathbf{h}^\mathrm{\scriptscriptstyle RU}$, lowering the NMSE compared to that with wider angular spread. 

Finally, we note that the NMSE shows no perceptible improvement with the MDFT training phase design as compared to the DFT training phases, as the error from unknown scattering in the BS-RIS channel is sufficiently large to obscure the differences between reasonable training phase designs for this range of $K$. When $K$ is very small, even the random approach has a similar NMSE as here the error from unknown scattering is very dominant.

\begin{figure}[ht]
	\centering
	\includegraphics[trim=1cm 0.04cm 1.6cm 0.1cm, clip=true, width=1\columnwidth]{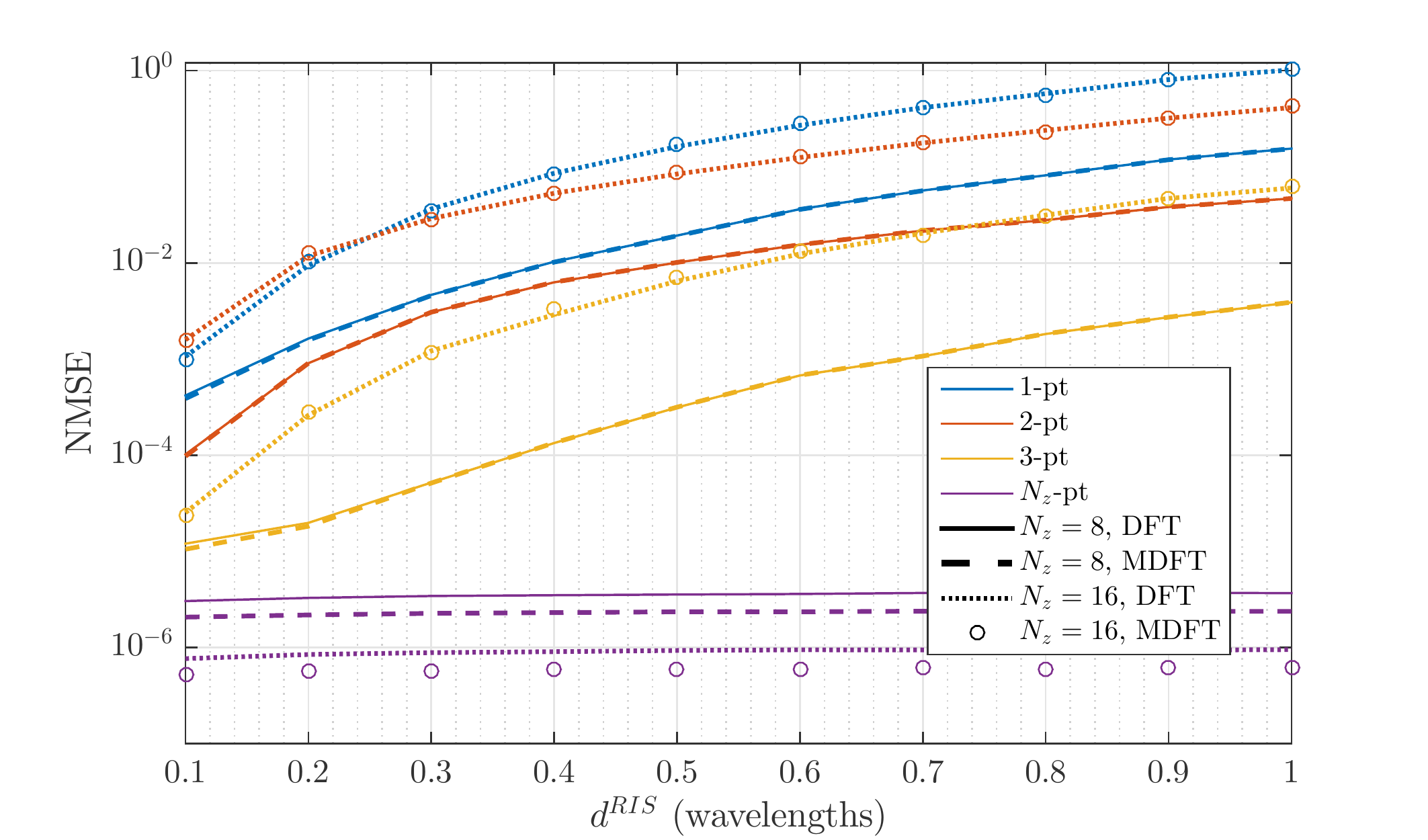}
	\caption{NMSE of $\hat{\mathbf{h}}_{\tau_1}^\mathrm{\scriptscriptstyle TOT}$ vs $d^\mathrm{\scriptscriptstyle RIS}$ under pure LoS conditions.}
	\label{fig:NMSE_vs_dRIS_LoS}
\end{figure}

\begin{figure}[ht]
	\centering
	\includegraphics[trim=1cm 0.04cm 1.6cm 0.1cm, clip=true, width=1\columnwidth]{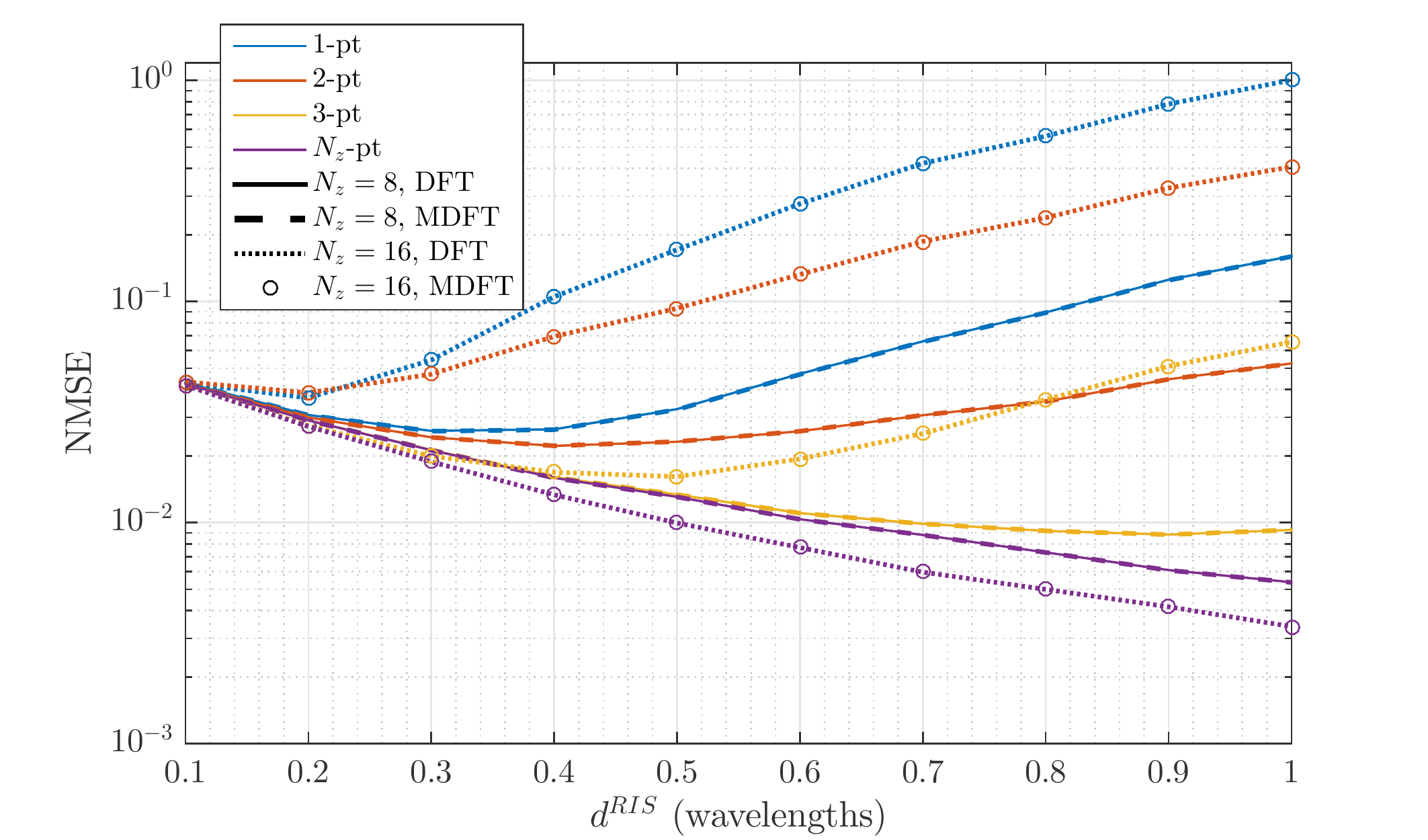}
	\caption{NMSE of $\hat{\mathbf{h}}_{\tau_1}^\mathrm{\scriptscriptstyle TOT}$ vs $d^\mathrm{\scriptscriptstyle RIS}$ with $K=12$dB.}
	\label{fig:NMSE_vs_dRIS_NLoS}
\end{figure}

Fig.~\ref{fig:NMSE_vs_dRIS_LoS} examines the effects of  RIS inter-element spacing and column size when $\mathbf{H}^\mathrm{\scriptscriptstyle BR}$ is pure LoS ($K\to\infty$). For all interpolating methods (1-pt, 2-pt and 3-pt) larger inter-element spacing and larger $N_z$ reduce the correlation in $\mathbf{h}^\mathrm{\scriptscriptstyle RU}$ \cite{antenna_config_paper} hindering interpolation and increasing the NMSE. Similarly, increasing levels of interpolation (from 3-pt to 1-pt) also increases the NMSE. These trends are also apparent in Fig.~\ref{fig:NMSE_vs_dRIS_NLoS} which considers a NLoS $\mathbf{H}^\mathrm{\scriptscriptstyle BR}$  channel with $K=12$dB. Note that without interpolation (the $N_z$-pt method) the NMSE is unaffected by the spacing as every channel element is being separately estimated. Mathematically, this can be seen in \eqref{A1} where ${A}_1$ is independent of spacing. The MDFT technique is seen to have better NMSE than DFT for the $N_z$-pt method, but this improvement is indiscernible for the interpolating methods as the interpolation errors obscure the MDFT improvements. Another trend which is different for the $N_z$-pt method is that larger $N_z$ reduces NMSE. This can be deduced from \eqref{A1} where the numerator of ${A}_1$ contains a weighted average of the error $\boldsymbol{\epsilon}^\mathrm{\scriptscriptstyle RU}$ and converges to zero as the number of RIS elements increases.

Fig.~\ref{fig:NMSE_vs_dRIS_NLoS} repeats the results in Fig.~\ref{fig:NMSE_vs_dRIS_LoS} for a NLoS $\mathbf{H}^\mathrm{\scriptscriptstyle BR}$  channel with $K=12$dB. Here, the benefits of MDFT are obscured by the errors introduced by the neglected scattering, even for the $N_z$-pt method. The other trend which is different in NLoS conditions is the relationship between NMSE and spacing. In Fig.~\ref{fig:NMSE_vs_dRIS_NLoS} we see that the NMSE of the $N_z$-pt method drops as spacing is increased (whereas it was constant in LoS). This is due to the  $A_2$ term in \eqref{eq:A2}. The inner product in $A_2$ decreases rapidly with increased inter-element spacing as this reduces the alignment between the LoS and scattered ray vectors, decreasing the NMSE for the $N_z$-pt method. For the interpolating methods, the NMSE initially drops and then increases with spacing. The initial drop is caused by the reduction in alignment between the LoS and scattered ray vectors in $A_2$ and the eventual rise is due to interpolation over less correlated channels.

Figs.~\ref{fig:NMSE_vs_dRIS_LoS}  and~\ref{fig:NMSE_vs_dRIS_NLoS} look at the size of channel estimation error. This is important, but the key question is how these errors influence the setting of the RIS phases. Hence, in Figs.~\ref{fig:phase_error_vs_SNRt_bestcase} and~\ref{fig:phase_error_vs_SNRt_worstcase} we examine the error in the RIS transmission phases calculated using the estimated channels, as compared to the ideal phase design obtained with perfect CSI. We plot the MS phase error, given by $\frac{1}{N}\sum_{i=1}^N(|\angle\bar{\boldsymbol{\Phi}}_{ii} - \angle\bar{\boldsymbol{\Phi}}_{ii}^{*}|^2)$,
where $\bar{\boldsymbol{\Phi}}$ is given in \eqref{eq:optimal_phases2}, $\bar{\boldsymbol{\Phi}}^{*}$ is calculated using \eqref{eq:optimal_phases}, and $\mathbf{a}^\mathrm{\scriptscriptstyle B}(\bar{\phi}^\mathrm{\scriptscriptstyle B},\bar{\theta}^\mathrm{\scriptscriptstyle B})$ and $\mathbf{a}^\mathrm{\scriptscriptstyle R}(\bar{\phi}^\mathrm{\scriptscriptstyle R},\bar{\theta}^\mathrm{\scriptscriptstyle R})$ are replaced with the left and right singular vectors of $\mathbf{H}^\mathrm{\scriptscriptstyle BR}$, respectively.
	
	Fig.~\ref{fig:phase_error_vs_SNRt_bestcase} plots the MS phase error vs the training SNR, $\rho$, for the parameters given in Scenario 1. In Scenario 1, the elements in the UE-RIS channel are highly correlated and the RIS-BS is near-perfectly estimated by the LoS component, providing a "best-case scenario" for the interpolating methods. As expected, the MS error drops with SNR and as finer interpolation is used. Fig.~\ref{fig:phase_error_vs_SNRt_bestcase} clearly demonstrates the benefits of MDFT which yields the most accurate RIS transmission phases with the greatest resilience to the training SNR. Substantial improvements over DFT and random are seen at low SNR. These differences are clear in Fig.~\ref{fig:phase_error_vs_SNRt_bestcase} but were small or indicernible in
	Figs.~\ref{fig:NMSE_vs_dRIS_LoS}  and~\ref{fig:NMSE_vs_dRIS_NLoS}. This is just a result of the metrics used. The maximum phase error resulting from an estimate of a channel element $\mathbf{h}_n$ with error $\boldsymbol{\epsilon}_n$ can be approximated by $\tan^{-1}(|\mathbf{\epsilon}_n|/|\mathbf{h}_n|)\approx|\mathbf{\epsilon}_n|/|\mathbf{h}_n|$ as often $|\mathbf{\epsilon}_n|/|\mathbf{h}_n|\ll 1$. Meanwhile, the NMSE is dependent on $(|\mathbf{\epsilon}_n|/|\mathbf{h}_n|)^2$, resulting in smaller values which obscure the differences between training phase designs.

	Fig.~\ref{fig:phase_error_vs_SNRt_worstcase} shows the same results for Scenario 2 which provides a more challenging situation for interpolation, with moderate scattering in the RIS-BS channel and parameters which decorrelate the elements of $\mathbf{h}^\mathrm{\scriptscriptstyle RU}$. All trends follow those in Fig.~\ref{fig:phase_error_vs_SNRt_bestcase}, however the MS error values are greatly increased. 	Fig.~\ref{fig:phase_error_vs_SNRt_worstcase} clearly demonstrates that, at reasonable values of $\rho$, the estimation accuracy is far more dependent on the channel conditions, RIS size, and level of interpolation used than on the training phase design. For example, in Scenario 1 at $\rho=5$dB, even random training matrices with the coarsest interpolation result in an RMS phase error under 1 radian. However, in Scenario 2, any interpolation coarser than the 3-pt method results in RMS phase errors approaching 2 radians even with the MDFT training phases and a very high training SNR.

\begin{figure}[ht]
	\centering
	\includegraphics[trim=1cm 0.04cm 1.6cm 0.1cm, clip=true, width=1\columnwidth]{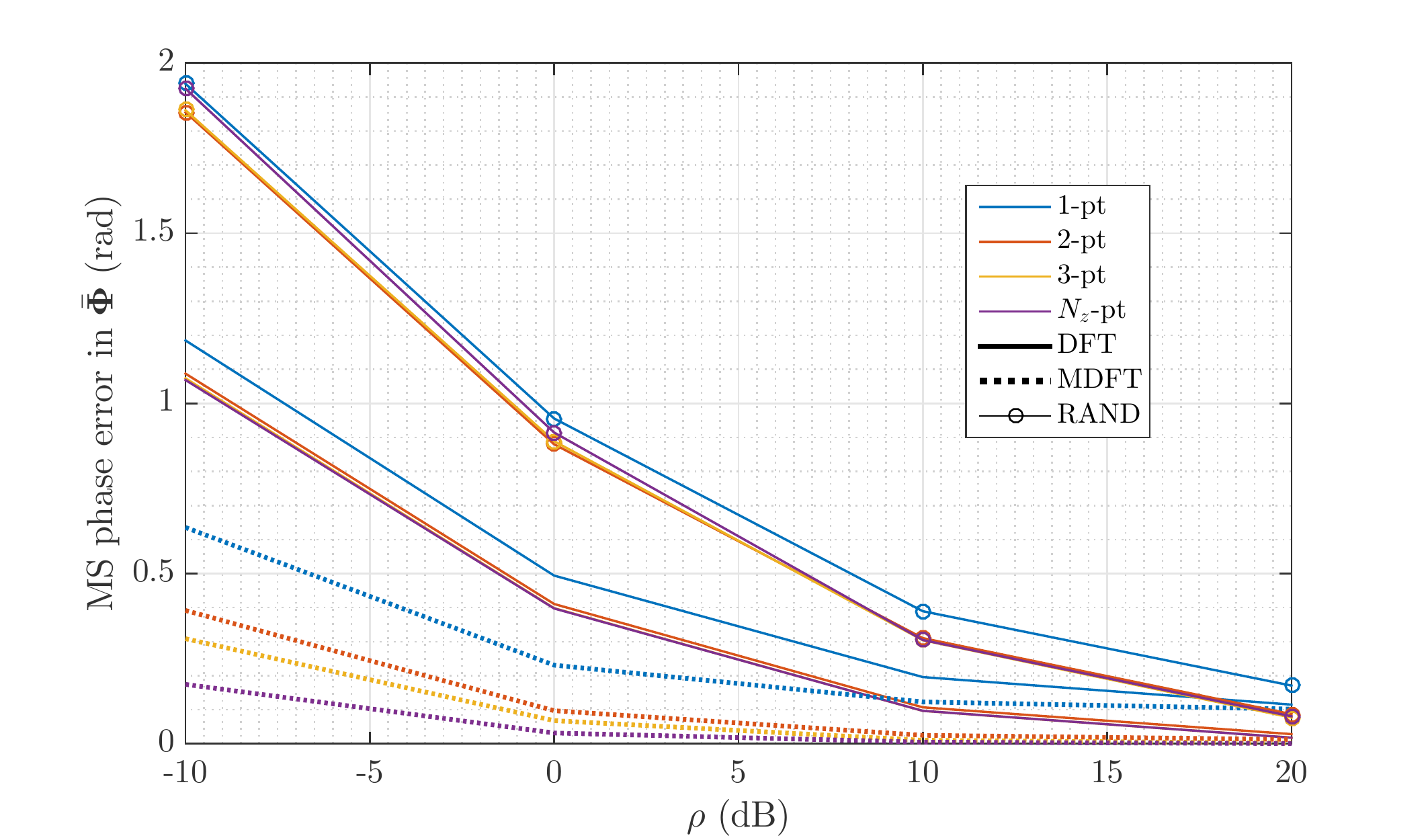}
	\caption{MS phase error in $\bar{\boldsymbol{\Phi}}$ vs uplink training SNR (Scenario 1).}
	\label{fig:phase_error_vs_SNRt_bestcase}
\end{figure}

\begin{figure}[ht]
	\centering
	\includegraphics[trim=1cm 0.04cm 1.6cm 0.1cm, clip=true, width=1\columnwidth]{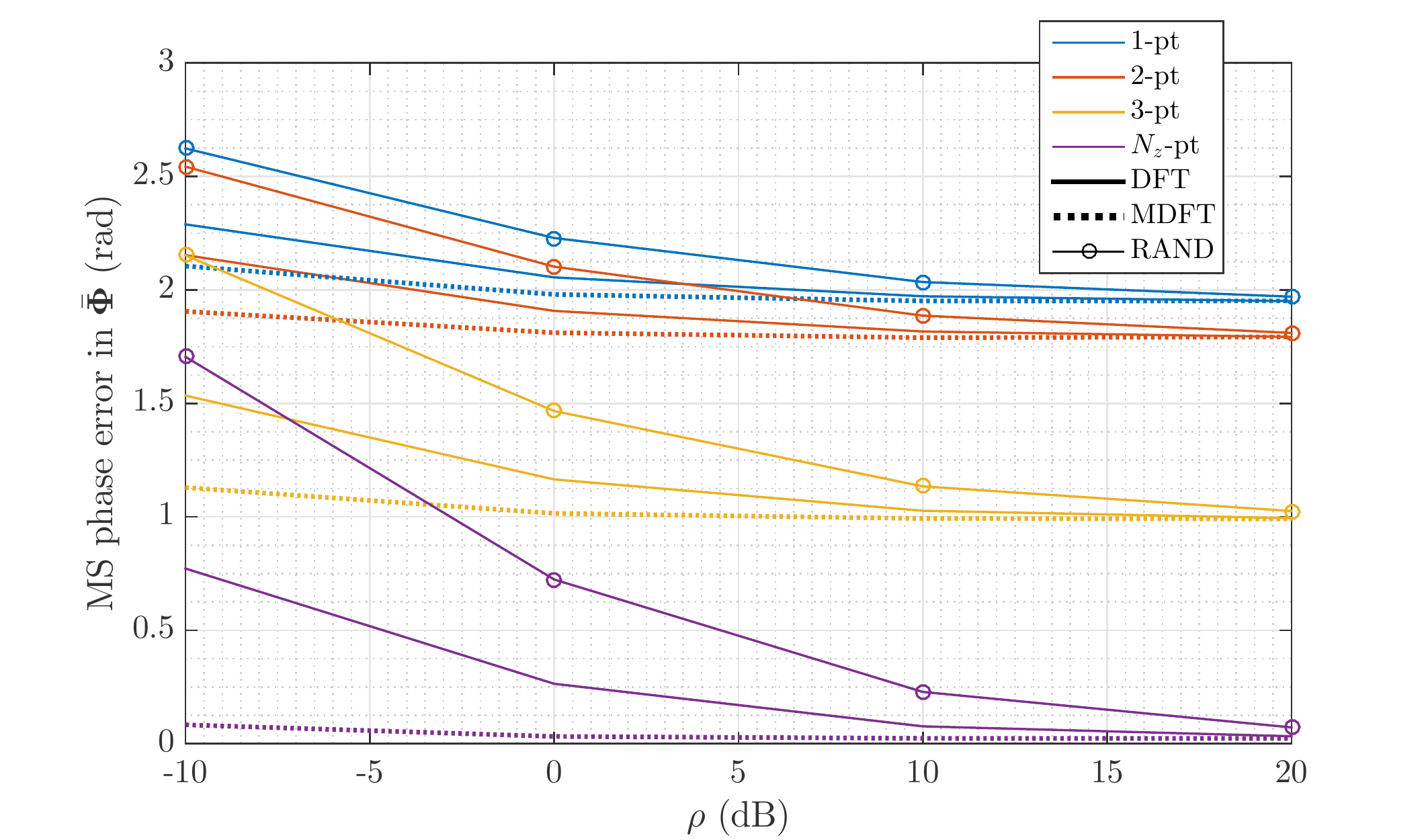}
	\caption{MS phase error in $\bar{\boldsymbol{\Phi}}$ vs uplink training SNR (Scenario 2).}
	\label{fig:phase_error_vs_SNRt_worstcase}
\end{figure}

To predict how the errors in Figs.~\ref{fig:phase_error_vs_SNRt_bestcase} and~\ref{fig:phase_error_vs_SNRt_worstcase} impact overall performance, we now examine the effects of the MS phase error in the RIS transmission phases on the resulting SE using the simplified analysis in Sec.~\ref{simplified}. To do so, in Fig.~\ref{fig:SE_vs_phase_error}, we plot the single-user SE with artificial Gaussian phase errors whose standard deviation is the RMS phase error (the square root of the MS values shown in Figs.~\ref{fig:phase_error_vs_SNRt_bestcase} and~\ref{fig:phase_error_vs_SNRt_worstcase}).
Fig.~\ref{fig:SE_vs_phase_error} plots the SE with perfect CSI against the synthesized RMS phase error for Scenarios 1 and 2, including the case for a blocked UE-BS channel, and the approximated lower bound on the performance calculated using \eqref{eq:SE_drop}. We note that the very simple SE loss predicted by \eqref{eq:SE_drop} is very accurate up to an RMS phase error of 1.5 radians.
\begin{figure}[ht]
	\centering
	\includegraphics[trim=1cm 0.04cm 1.6cm 0.1cm, clip=true, width=1\columnwidth]{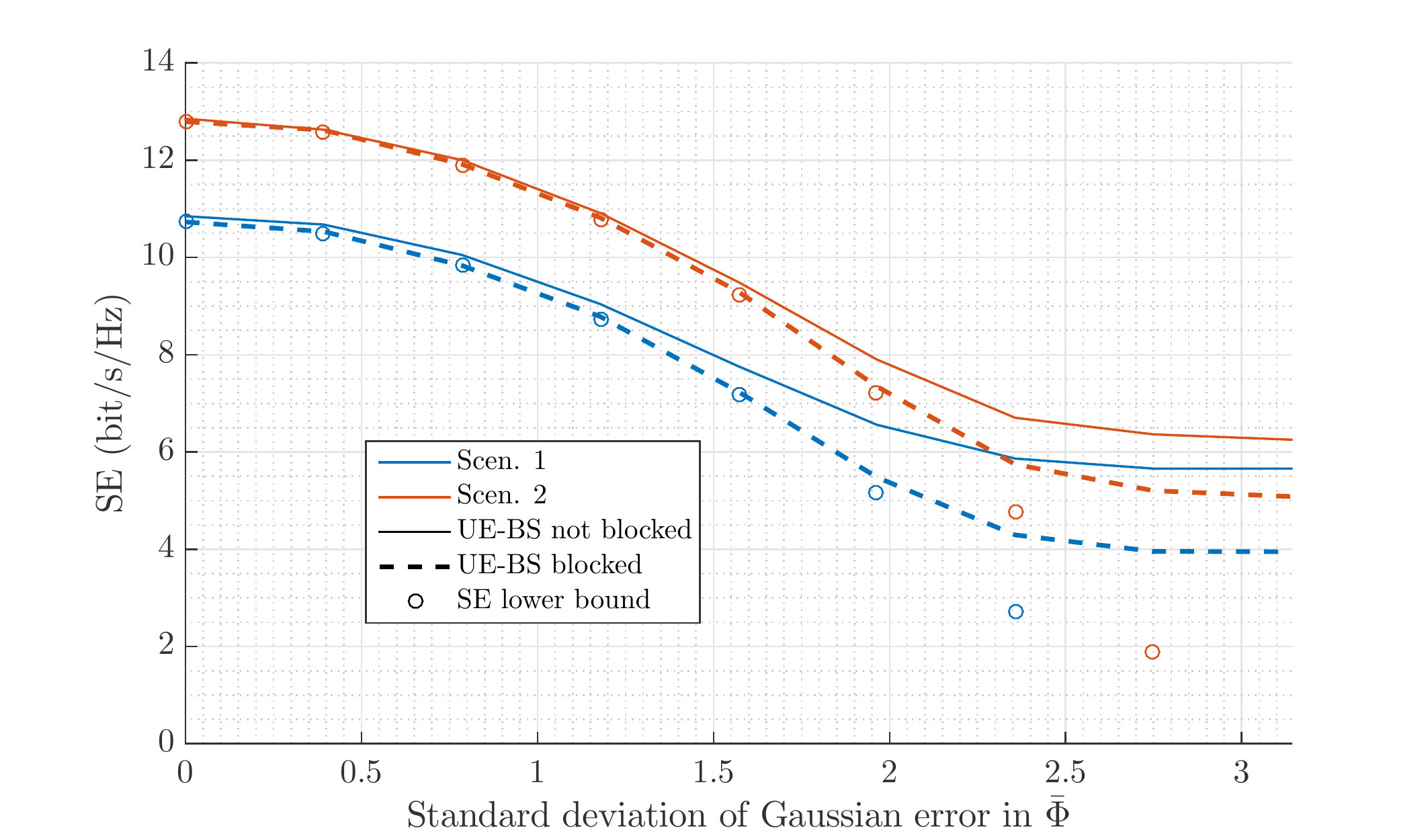}
	\caption{SE vs standard deviation of Gaussian error in $\bar{\boldsymbol{\Phi}}$, using single-user RIS phase design from \cite{Singh2020Optimal} and MRC processing.}
	\label{fig:SE_vs_phase_error}
\end{figure}
Note that for both scenarios the SE loss for an RMS error of 1 radian or 57$^\circ$ in $\bar{\boldsymbol{\Phi}}$ is around 13\%. Hence, the simulations and analysis show remarkable resilience to loss of accuracy of the phases despite the apparent wide variation in the RMS phase errors in Figs.~\ref{fig:phase_error_vs_SNRt_bestcase} and~\ref{fig:phase_error_vs_SNRt_worstcase}.   From Fig.~\ref{fig:phase_error_vs_SNRt_bestcase}, we see that RMS phases less than one are achieved by all methods when $\rho>0$dB. In contrast, for Scenario 2, Fig.~\ref{fig:phase_error_vs_SNRt_worstcase} shows that the use of MDFT and at least three elements per 16-column RIS are required to achieve an RMS value around 1 radian.

Finally, in Figs.~\ref{fig:SE_vs_tau2_Nz8} and~\ref{fig:SE_vs_tau2_Nz16} we examine the SE with MRC processing following re-estimation of $\hat{\mathbf{h}}_{\tau_1}^\mathrm{\scriptscriptstyle TOT}$ with $\tau_2$ additional symbols in Stage 2. We take into account the portion of usable frame symbols remaining for data transmission following the training procedure, using \eqref{eq:SE} with a frame length of $T=400$ symbols. Fig.~\ref{fig:SE_vs_tau2_Nz8} plots the SE for Scenarios 1 and 2 with $N_z=8$, while Fig.~\ref{fig:SE_vs_tau2_Nz16} shows the same for $N_z=16$. All results clearly show that $\tau_2=1$ is sufficient for the re-estimation stage.
\begin{figure}[ht]
	\centering
	\includegraphics[trim=1cm 0.04cm 1.6cm 0.1cm, clip=true, width=1\columnwidth]{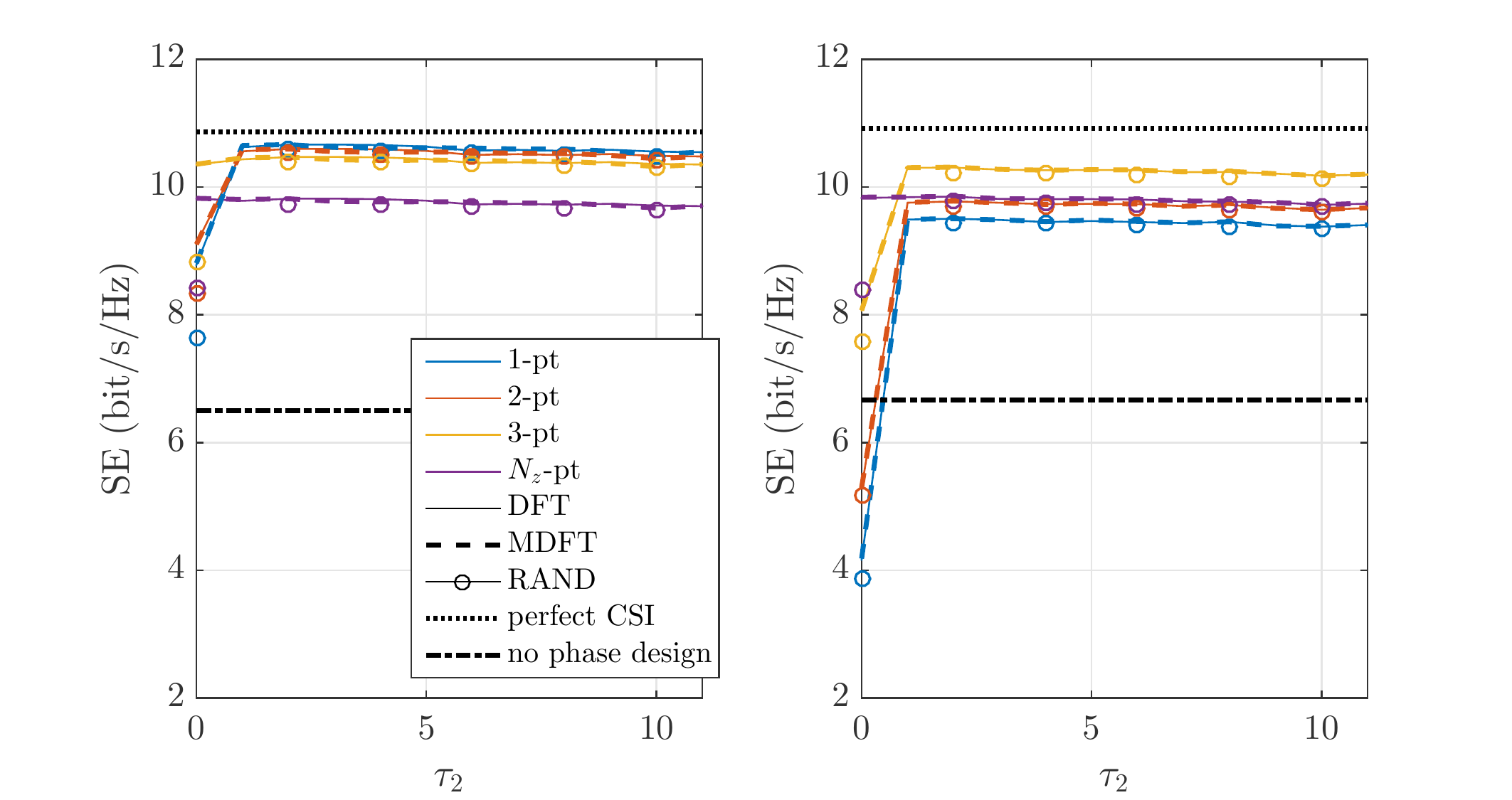}
	\caption{UL SE vs $\tau_2$, using single-user RIS phase design from \cite{Singh2020Optimal} and MRC processing, for Scenario 1 (left) and Scenario 2 (right) with $N_z=8$.}
	\label{fig:SE_vs_tau2_Nz8}
\end{figure}
As previously deduced, the the interpolating methods perform extremely well under Scenario 1. Particularly when re-estimating $\hat{\mathbf{h}}^\mathrm{\scriptscriptstyle}$ in Stage 2, even the coarsest interpolation method with any training phase design gives near-optimal performance for $N_z=8$ and $N_z=16$. Here, the rate loss due interpolation is less than the gain due to the use of extra data symbols. Hence, when the channels are highly correlated, we can afford to greatly reduce the number of training symbols required, increasing the usable data symbols in the frame.
\begin{figure}[ht]
	\centering
	\includegraphics[trim=1cm 0.04cm 1.6cm 0.1cm, clip=true, width=1\columnwidth]{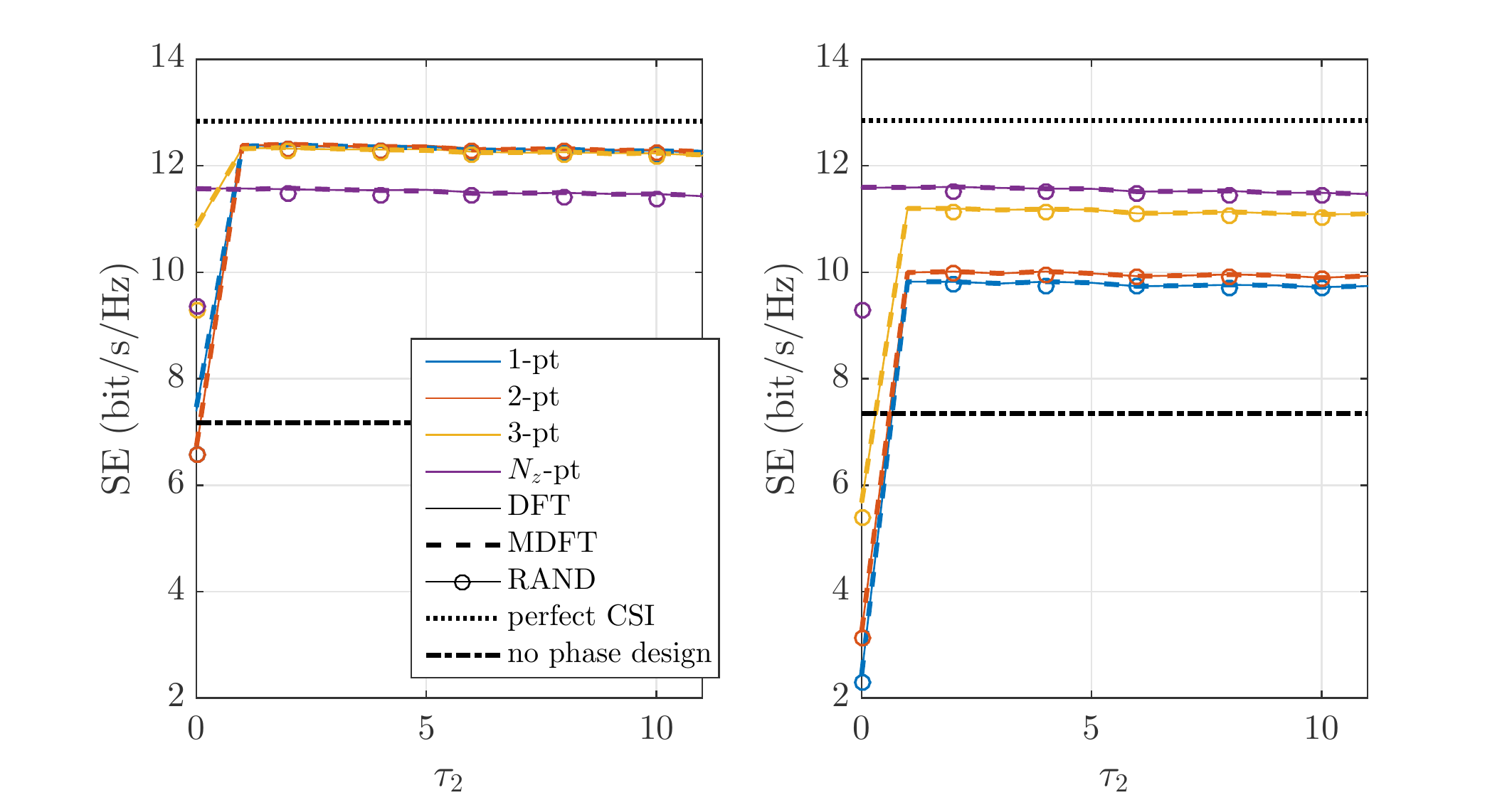}
	\caption{UL SE vs $\tau_2$, using single-user RIS phase design from \cite{Singh2020Optimal} and MRC processing, for scenario 1 (left) and scenario 2 (right) with $N_z=16$.}
	\label{fig:SE_vs_tau2_Nz16}
\end{figure}
Under Scenario 2, the coarser interpolation methods suffer as they cause sufficient errors in the RIS transmission phases to harm performance despite the greater usable frame size they afford. This is seen in  Fig.~\ref{fig:SE_vs_tau2_Nz8}  where the 3-pt method  provides better performance than the $N_z$-pt method when $N_z=8$. In Fig.~\ref{fig:SE_vs_tau2_Nz16},  $N_z=16$ and interpolation is more challenging. Here, all interpolation methods corrupt the RIS transmission phases significantly to counteract the benefits from the increased usable frame symbols afforded by these methods. Here, 3-pt or $N_z$-pt interpolation is beneficial. Also shown in Figs.~\ref{fig:SE_vs_tau2_Nz8} and~\ref{fig:SE_vs_tau2_Nz16}  is the curve labeled {\textit{no phase design}}, which corresponds to no channel estimation ($\tau_1=\tau_2=0$) and the use of random RIS phases during transmission. This benchmark result shows that a large percentage of the possible improvement in SE is achieved by a mixture of interpolation and MDFT which greatly accelerates channel estimation.

Clearly, the estimation protocol is also extremely robust to the training phase design, giving similar average performance for any training phase matrix, provided the total channel is re-estimated in Stage 2 with only one additional training symbol. This is explained by Figs.~\ref{fig:phase_error_vs_SNRt_bestcase}-\ref{fig:SE_vs_phase_error} where the MS phase error variation across techniques at $\rho=5$dB is too small to have a significant SE effect. Note that with these parameter settings, although the random training phase design performs well on average, it can deliver much lower SE values than DFT or MDFT.  Hence, a fixed training matrix is much more reliable and, as the MDFT yields equal or better performance to the DFT for all metrics considered, it is logical to use the MDFT design.

\section{Conclusion}
We provide an exact error analysis of the two-stage channel estimation procedure in \cite{miller2021efficient} leading to an optimal training phase design based on the MDFT matrix. In conjunction with substantial channel interpolation, this technique is shown to accelerate channel estimation considerably when the RIS-BS channel is near LoS, solving the problem that prior techniques were only efficient for full-rank RIS-BS channels. SE results from the new design are shown to be surprisingly resilient to the training SNR, and this property is explained by a simplified SNR analysis. Results reveal that, while the MDFT design is the most reliable and robust, the accuracy of estimation is more dependent on the channel correlation and RIS size than the training phase design. For less correlated channels with a large RIS, the interpolating methods  degrade the accuracy in the RIS transmission phases sufficiently to nullify any potential benefit from the increased portion of usable frame symbols. However, the overall channel estimation protocol provides extremely robust performance for all conditions considered, provided one additional symbol is used for re-estimation in Stage 2.
\begin{appendices}
\section{Proof of Theorem 1}\label{App}
From \eqref{eq:X11_general} and \eqref{eq:X22_general} the traces of $\mathbf{X}_{11}$ and $\mathbf{X}_{22}$ are given by
\begin{align*}
\mathrm{tr}(\mathbf{X}_{11}) &=(M\bar{\beta}^\mathrm{\scriptscriptstyle BR})^{-1}\big[\mathrm{tr}\big(\boldsymbol{\Omega}^{-1}\big)
+\alpha^{-1}\boldsymbol{\omega}^\dagger\boldsymbol{\Omega}^{-2}\boldsymbol{\omega}\big]
\end{align*}
and
\begin{align*}
\mathrm{tr}(\mathbf{X}_{22}) &= (\tau_1)^{-1}\big[M
+\alpha^{-1}\boldsymbol{\omega}^\dagger\boldsymbol{\Omega}^{-1}\boldsymbol{\omega}\big].
\end{align*}
Substituting \eqref{eq:ABD} into \eqref{eq:X_parts_general}, we see that $\mathbf{X}_{11}$ involves the term $[\boldsymbol{\Omega} - \tau_1^{-1}\boldsymbol{\omega}\boldsymbol{\omega}]^{-1}$. Noting that $\mathbf{X}_{11}$ must be positive definite, we require that $|\boldsymbol{\Omega}|(1 - \tau_1^{-1}\boldsymbol{\omega}\boldsymbol{\Omega}^{-1}\boldsymbol{\omega})\geq0$. For $\boldsymbol{\Omega}$ to be invertible, we must have $|\boldsymbol{\Omega}|\geq 0$, therefore $\alpha\geq0$.
Furthermore, $\boldsymbol{\omega}^\dagger\boldsymbol{\Omega}^{-2}\boldsymbol{\omega}$ is quadratic in form, hence the diagonal elements are also positive. These observations show that all the terms in $\mathrm{tr}(\mathbf{X}_{11})$ and $\mathrm{tr}(\mathbf{X}_{22})$ are positive. Hence,  minimizing the trace of $\mathbf{X}$ requires training phases which yield $\boldsymbol{\omega} = \mathbf{0}^{N\times1}$ and minimize $\mathrm{tr}(\boldsymbol{\Omega}^{-1})$. The former condition would also minimize $\mathrm{tr}(\mathbf{Y}_{22}\mathbf{Y}_{22}^\dagger)$ using \eqref{eq:Y22_3}. 

Since, the MDFT approach satisfies $\boldsymbol{\omega}^{\textrm{MD}} = \mathbf{0}^{N\times1}$, it only remains to show that it also minimizes $\mathrm{tr}(\boldsymbol{\Omega}^{-1})$. Let $\lambda_1, \ldots , \lambda_{N}$ be the eigenvalues of an $N \times N$ matrix, $\mathbf{P}$. Then, the Cauchy-Schwartz inequality gives:
\begin{equation}\label{CS}
N=\sum_{i=1}^N \sqrt{\lambda_i}\sqrt{\lambda_i^{-1}} \leq \sqrt{\sum_{i=1}^N \lambda_i \sum_{i=1}^N \lambda_i^{-1}}.
\end{equation}
From \eqref{CS}, we obtain the trace inequality $\mathrm{tr}(\mathbf{P}^{-1}) \geq N^2/\mathrm{tr}(\mathbf{P})$. Now, for all training phase matrices we have:
\begin{equation}\label{trace_eq}
\mathrm{tr}(\boldsymbol{\Omega})=\mathrm{tr}\left(\sum_{t=1}^{\tau_1} \boldsymbol{\psi}_t^*\boldsymbol{\psi}_t^T\right)=\sum_{t=1}^{\tau_1}\mathrm{tr}\left( \boldsymbol{\psi}_t^T\boldsymbol{\psi}_t^*\right)=\tau_1 N.
\end{equation}
Substituting \eqref{trace_eq} into the trace inequality gives the lower bound  $\mathrm{tr}(\boldsymbol{\Omega}^{-1}) \geq N/\tau_1$. Since $\boldsymbol{\Omega}^{\textrm{MD}}$ satisfies the lower bound,  $\mathrm{tr}(\boldsymbol{\Omega}^{-1})= N/\tau_1$, the proof is complete.

\end{appendices}
	\bibliographystyle{IEEEtran}
	\bibliography{RIS_CE_mod_dft}

\begin{thebibliography}{10}
\providecommand{\url}[1]{#1}
\csname url@samestyle\endcsname
\providecommand{\newblock}{\relax}
\providecommand{\bibinfo}[2]{#2}
\providecommand{\BIBentrySTDinterwordspacing}{\spaceskip=0pt\relax}
\providecommand{\BIBentryALTinterwordstretchfactor}{4}
\providecommand{\BIBentryALTinterwordspacing}{\spaceskip=\fontdimen2\font plus
\BIBentryALTinterwordstretchfactor\fontdimen3\font minus
  \fontdimen4\font\relax}
\providecommand{\BIBforeignlanguage}[2]{{%
\expandafter\ifx\csname l@#1\endcsname\relax
\typeout{** WARNING: IEEEtran.bst: No hyphenation pattern has been}%
\typeout{** loaded for the language `#1'. Using the pattern for}%
\typeout{** the default language instead.}%
\else
\language=\csname l@#1\endcsname
\fi
#2}}
\providecommand{\BIBdecl}{\relax}
\BIBdecl

\bibitem{Mishra2019Channel}
D.~Mishra and H.~Johansson, ``{Channel estimation and low-complexity
  beamforming design for passive intelligent surface assisted {MISO} wireless
  energy transfer},'' \emph{Proc. {IEEE} ICASSP}, pp. 4659--4663, May 2019.

\bibitem{Nadeem2020Intelligent}
Q.-U.-A. Nadeem, H.~Alwazani, A.~Kammoun, A.~Chaaban, M.~Debbah \emph{et~al.},
  ``Intelligent reflecting surface-assisted multi-user {MISO} communication:
  Channel estimation and beamforming design,'' \emph{IEEE Open J. Commun.
  Soc.}, vol.~1, pp. 661--680, May 2020.

\bibitem{Bjornson2020Intelligent}
E.~{Björnson}, {\"{O}}.~{Özdogan}, and E.~G. {Larsson}, ``Intelligent
  reflecting surface versus decode-and-forward: How large surfaces are needed
  to beat relaying?'' \emph{{IEEE} Wireless Commun. Lett.}, vol.~9, no.~2, pp.
  244--248, Feb. 2020.

\bibitem{Jensen2020Optimal}
T.~L. Jensen and E.~{De Carvalho}, ``An optimal channel estimation scheme for
  intelligent reflecting surfaces based on a minimum variance unbiased
  estimator,'' \emph{Proc. {IEEE} ICASSP}, pp. 5000--5004, May 2020.

\bibitem{Wan2020Broadband}
Z.~{Wan}, Z.~{Gao}, and M.~{Alouini}, ``Broadband channel estimation for
  intelligent reflecting surface aided {mmWave} massive {MIMO} systems,''
  \emph{Proc. {IEEE} ICC}, pp. 1--6, June 2020.

\bibitem{Zheng2020Intelligent}
B.~Zheng, C.~You, and R.~Zhang, ``Intelligent reflecting surface assisted
  multi-user {OFDMA}: Channel estimation and training design,'' \emph{{IEEE}
  Trans. Wireless Commun.}, vol.~19, no.~12, pp. 8315--8329, Mar. 2020.

\bibitem{Singh2020Optimal}
I.~Singh, P.~J. Smith, and P.~A. Dmochowski, ``Optimal {SNR} analysis for
  single-user {RIS} systems,'' \emph{Proc. {IEEE} PIMRC}, pp. 1--6, Sep. 2021.

\bibitem{Hu2019TwoTimescale}
C.~Hu, L.~Dai, S.~Han, and X.~Wang, ``Two-timescale channel estimation for
  reconfigurable intelligent surface aided wireless communications,''
  \emph{{IEEE} Trans. Commun.}, pp. 1--1, 2021.

\bibitem{Liu2019MatrixCalibration}
H.~Liu, X.~Yuan, and Y.-J.~A. Zhang, ``Matrix-calibration-based cascaded
  channel estimation for reconfigurable intelligent surface assisted multiuser
  {MIMO},'' \emph{{IEEE} J. Sel. Areas Commun.}, vol.~38, no.~11, pp.
  2621--2636, Dec. 2019.

\bibitem{Nadeem2020asymptotic}
Q.~U.~A. Nadeem, A.~Kammoun, A.~Chaaban, M.~Debbah, and M.~S. Alouini,
  ``{Asymptotic max-min SINR analysis of reconfigurable intelligent surface
  assisted MISO systems},'' \emph{{IEEE} Trans. Wireless Commun.}, vol.~19,
  no.~12, pp. 7748--7764, Dec. 2020.

\bibitem{You2020Intelligent}
C.~You, B.~Zheng, and R.~Zhang, ``Intelligent reflecting surface with discrete
  phase shifts: Channel estimation and passive beamforming,'' \emph{Proc.
  {IEEE} ICC}, pp. 1--6, Jun. 2020.

\bibitem{miller2021efficient}
C.~Miller, P.~A. Dmochowski, and P.~J. Smith, ``Efficient channel estimation
  for {RIS},'' in \emph{Proc. {IEEE} ICC}, Jun. 2021, pp. 1--6.

\bibitem{wang2019channel}
A.~{Wang}, R.~{Yin}, and C.~{Zhong}, ``Channel estimation for uniform
  rectangular array based massive {MIMO} systems with low complexity,''
  \emph{{IEEE} Trans. Veh. Technol.}, vol.~68, no.~3, pp. 2545--2556, Mar.
  2019.

\bibitem{sangodoyin_cluster_2018}
S.~Sangodoyin, V.~Kristem, C.~U. Bas, M.~Käske, J.~Lee \emph{et~al.},
  ``Cluster characterization of 3-{D} {MIMO} propagation channel in an urban
  macrocellular environment,'' \emph{{IEEE} Trans. Wireless Commun.}, vol.~17,
  no.~8, pp. 5076--5091, Aug. 2018.

\bibitem{3GPP}
3GPP, ``{Study on channel model for frequencies from 0.5 to 100 {GH}z},'' 3rd
  Generation Partnership Project (3GPP), Tech. Rep. {TR} 38.901 (V15.1.0), Aug.
  2018.

\bibitem{antenna_config_paper}
C.~L. {Miller}, P.~J. {Smith}, and P.~A. {Dmochowski}, ``Space-constrained
  arrays for massive {MIMO},'' \emph{{IEEE} Wireless Commun. Lett.}, pp. 1--1,
  2021.

\bibitem{horn_johnson}
R.~A. Horn and C.~R. Johnson, \emph{Matrix {A}nalysis}.\hskip 1em plus 0.5em
  minus 0.4em\relax Cambridge University Press, 2012.

\end{thebibliography}
	%
\end{document}